\documentclass[journal,12pt,onecolumn,draftclsnofoot,]{IEEEtran}

\usepackage[latin1]{inputenc}
\usepackage{times,amsmath}
\usepackage{amsfonts}
\usepackage{pstool}
\usepackage{subfigure}
\usepackage{multirow}
\usepackage{enumerate}
\usepackage{graphicx}
\usepackage{MnSymbol}
\usepackage{stfloats}
\usepackage[table]{xcolor}
\usepackage[square, comma, sort&compress, numbers]{natbib}
\usepackage{nohyperref}
\usepackage{algorithm,algorithmic}
\usepackage{bigints}
\usepackage{amsmath}
\usepackage{amsthm}

\newtheorem*{remark}{Remark}
\newtheorem{prop}{Proposition}
\newcounter{tempEquationCounter}
\newcounter{thisEquationNumber}

\DeclareMathOperator{\E}{\mathbb{E}}

\makeatletter
\newcommand{\vast}{\bBigg@{4}}
\newcommand{\Vast}{\bBigg@{5}}
\makeatother

\graphicspath{ {Figures/} }

\newcounter{relctr} %% <- counter for relations
\everydisplay\expandafter{\the\everydisplay\setcounter{relctr}{0}} %% <- reset every eq
\newcommand\labelrel[2]{%
  \begingroup
    \refstepcounter{relctr}%
    \stackrel{\textnormal{(\alph{relctr})}}{\mathstrut{#1}}%
    \originallabel{#2}%
  \endgroup
}
\AtBeginDocument{\let\originallabel\label} %% <- store original definition

\begin{document}

\title{Statistical QoS Analysis of Reconfigurable Intelligent Surface-assisted D2D Communication}

\author{
\IEEEauthorblockN{Syed Waqas Haider Shah, \textit{Student Member, IEEE}, Adnan Noor Mian, \textit{Member, IEEE}, Shahid Mumtaz, \textit{Senior Member, IEEE}, Anwer Al-Dulaimi, \textit{Senior Member, IEEE}, Chih-Lin I, \textit{Fellow, IEEE}, and Jon Crowcroft, \textit{Fellow, IEEE}}
}

\maketitle

\long\def\symbolfootnote[#1]#2{\begingroup%
\def\thefootnote{\fnsymbol{footnote}}\footnote[#1]{#2}\endgroup}
\symbolfootnote[0]{\hrulefill \\
Copyright (c) 2015 IEEE. Personal use of this material is permitted. However, permission to use this material for any other purposes must be obtained from the IEEE by sending a request to pubs-permissions@ieee.org.

The work of Syed Waqas Haider Shah is supported by the Commonwealth Scholarship Commission, funded by the U.K. Government. The work of Shahid Mumtaz is supported by the FCT project (Intelligent and Sustainable Aerial-Terrestrial IoT Networks-BATS) PTDC/EEI-TEL/1744/2021.

Syed Waqas Haider Shah and Jon Crowcroft are with the Computer Lab, University of Cambridge, 15 JJ Thomson Avenue, Cambridge, UK CB3 0FD (\{sw920, jon.crowcroft\}@cl.cam.ac.uk). Syed Waqas Haider Shah is also with the Electrical Engineering Department, Information Technology University, Lahore 54000, Pakistan (waqas.haider@itu.edu.pk).

Adnan Noor Mian is with the Computer Science Department, Information Technology University, Lahore 54000, Pakistan (adnan.noor@itu.edu.pk).

Shahid Mumtaz is with the Institute of Telecommunication, University of Aveiro, 4554 Aveiro, Portugal
(smumtaz@av.it.pt).

Anwer Al-Dulaimi is with the R\&D Department, EXFO, Montreal, QC H4S 0A4, Canada
(anwer.al-dulaimi@exfo.com) 

Chih-Lin I is with China Mobile Research Institute, Beijing 100053, China
(icl@chinamobile.com)
}

\maketitle

\begin{abstract}
This work performs the statistical QoS analysis of a Rician block-fading reconfigurable intelligent surface (RIS)-assisted D2D link in which the transmit node operates under delay QoS constraints. First, we perform mode selection for the D2D link, in which the D2D pair can either communicate directly by relaying data from RISs or through a base station (BS). Next, we provide closed-form expressions for the effective capacity (EC) of the RIS-assisted D2D link. When channel state information at the transmitter (CSIT) is available, the transmit D2D node communicates with the variable rate $r_t(n)$ (adjustable according to the channel conditions); otherwise, it uses a fixed rate $r_t$. It allows us to model the RIS-assisted D2D link as a Markov system in both cases. We also extend our analysis to overlay and underlay D2D settings. To improve the throughput of the RIS-assisted D2D link when CSIT is unknown, we use the HARQ retransmission scheme and provide the EC analysis of the HARQ-enabled RIS-assisted D2D link. Finally, simulation results demonstrate that: i) the EC increases with an increase in RIS elements, ii) the EC decreases when strict QoS constraints are imposed at the transmit node, iii) the EC decreases with an increase in the variance of the path loss estimation error, iv) the EC increases with an increase in the probability of ON states, v) EC increases by using HARQ when CSIT is unknown, and it can reach up to $5\times$ the usual EC (with no HARQ and without CSIT) by using the optimal number of retransmissions. 

\end{abstract}
\IEEEpeerreviewmaketitle


\section{Introduction}
In recent years, cellular networks have seen exponential growth in popularity due to their ubiquitous and high data rate connectivity. Fifth-generation (5G) cellular networks, which are considered a game-changer for wireless connectivity, have started being deployed in various countries around the globe. The stories of its early deployment are overwhelming, and researchers are actively looking for other deployment scenarios to achieve its full breadth \cite{ericson2020}. At the same time, research on beyond 5G cellular networks such as sixth-generation (6G) cellular networks has also begun. 6G is expected to provide extreme data rates and ultra-reliability to a massive number of cellular devices. To provide seamless connectivity to billions of cellular devices, various enabling technologies are under consideration. Device-to-Device (D2D) communication is one of these enabling technologies that can provide low latency and energy-efficient connectivity without exerting extra burden on the cellular network infrastructure. D2D communication is a paradigm in which two closely located devices communicate without routing their data through a base station (BS). D2D devices are generally low-cost user equipments with limited resources and computational capacity. Therefore, successful D2D communication critically depends on the propagation environment. For instance, D2D communication may not be successful if the distance between the candidate D2D devices becomes too large or a line-of-sight link is not available. Moreover, resources for D2D users are allocated in either overlay or underlay settings. In the former, D2D and cellular users are assigned orthogonal physical resource blocks (PRBs), which ensures no inter-channel interference, whereas, in the latter, D2D users reuse the cellular users' PRBs, which reduces the spectrum scarcity issue of the cellular networks. However, in underlay settings, both the D2D and cellular users experience inter-channel interference that leads to poor communication quality. It is because the D2D users operate under interference power constraints while in underlay communication mode. This constraint limits the transmit power of the D2D users so that the cellular users who are using the same PRBs for transmission experience minimum inter-channel interference \cite{ji2020secrecy}. It is one of the main limiting factors on the achievable capacity of the D2D users and consequently on the communication quality \cite{kusaladharma2016underlay}. These reasons may restrict the applications and use cases of most existing approaches for D2D communication \cite{duan2020emerging}. To this end, Reconfigurable intelligent surface (RIS) \cite{basar2019wireless} is envisioned to further improve the performance of D2D communication link by passive beamforming, which will not only resolve the performance degradation of D2D communication due to interference while in underlay communication mode but can also open new horizons for smooth integration of D2D communication in various new applications and use cases. This can be achieved by configuring the RIS to control the reflection, refraction, and scattering of electromagnetic waves that impinge on the surface. Moreover, RISs allow transmission distance extension and energy-efficient communication, which are some of the core requirements of D2D communication in order to make it suitable for various use cases and applications \cite{chen2020reconfigurable}. In short, RIS-enabled D2D communication can help reduce the interference in the network, leading to higher data rates and better QoS performance for both D2D and cellular users.

RIS is a two-dimensional surface comprised of multiple scattering elements (called unit cells) and small electronic circuits with reconfigurable phase and magnitude responses. The unit cells are low cost and passive elements that can reflect the incoming signals to the desired receiver by adjusting their phase and amplitude response by appropriately configuring the low-power electronic circuits by using an RIS controller. RISs are expected to make the wireless environment controllable and reconfigurable, which is uncontrollable in current wireless systems due to the presence of non-configurable physical objects that influence the propagation of the signals. Since, the elements used for reflecting the signals are passive and the electronic circuits that ensure the tunability of the RIS consume almost no power once configured, RISs are envisioned to require smaller power than other technologies. This makes RISs an energy-efficient solution to enhancing communication quality, which is one of the key requirements of future cellular networks. Using the RIS for D2D communication can help reduce the interference in the network, leading to higher data rates and better quality-of-service (QoS) for both D2D and cellular users. 

Due to its simple and efficient operation, RISs have gained significant research attention recently. Researchers have proposed its possible integration with various wireless technologies such as millimeter-wave communication \cite{qiao2020secure}, cognitive radios \cite{zhang2020intelligent}, non-orthogonal multiple access \cite{mu2020exploiting}, Terahertz communication, beyond multi-input multi-output (MIMO) networks for 6G \cite{8981888,Nemanja2020},  wireless power transfer \cite{pan2020intelligent}, and unmanned aerial vehicle-based communication \cite{yang2020performance}. More recently, authors of \cite{chen2020reconfigurable} proposed the integration of RISs with D2D communication. The authors proposed an RIS phase shift optimization problem to maximize the system sum-rate under phase shift and power constraints. The authors of \cite{ji2020reconfigurable} also provide a joint optimization problem to maximize the sum-rate of cellular and D2D users. The authors of \cite{cai2020two} introduce a two-timescale optimization scheme for RIS-assisted D2D communication. The idea of this optimization scheme is to maximize the ergodic capacity of D2D users while keeping the outage probability of cellular users as a constraint. These studies aim to mitigate D2D interference by using RISs to enhance the communication capacity and quality of cellular users. In another study \cite{jia2020reconfigurable}, the authors investigate the use of RIS to improve the energy efficiency of D2D communication. The authors propose a joint power control and an RIS passive beamforming optimization algorithm for achieving energy-efficient D2D communication. Although these studies provide significant insights for RIS-assisted D2D communication, none of them study the impact of RISs on the QoS and reliability of D2D communication. An RIS-assisted D2D network must provide reliable and QoS-enabled communication for its potential use in various applications. To this end, we provide a throughput analysis of RIS-assisted D2D communication under statistical QoS constraints. For our analysis, we use a well-known analytical tool: the Effective Capacity (EC).

The EC is a link-layer model that uses the concepts of queuing theory to find the throughput of the system under statistical QoS constraints imposed at the transmitter's queue \cite{wu2003effective}. It provides a maximum arrival rate at the transmitter's queue while considering a time-varying channel. Due to its efficacy in finding the system's throughput under queuing constraints, it has been widely used in performance analysis of different wireless networks \cite{qiao2012effective,musavian2010effective,cheng2013qos}. It has also been used to find the statistical QoS guarantees for D2D communication in various network settings \cite{shah2019impact,shah2020statistical,shah2021effective}. In this paper, we provide, for the first time, the statistical QoS analysis and QoS provisioning of RIS-assisted D2D communication, by using the EC tool. This analysis will pave the way for possible integration of RIS-assisted D2D communication in beyond-5G cellular networks. Major contributions of this work are as follows:
\begin{itemize}
    \item We propose a mode selection mechanism for RIS-assisted D2D communication, which selects a communication link for uplink transmission from the candidate links (RIS-assisted D2D and RIS-assisted cellular links) based on path loss measurements of the respective links. We also provide a framework for calculating the path loss of the candidate RIS-assisted communication links. 
    \item We provide a throughput analysis of RIS-assisted D2D communication under statistical QoS guarantees imposed as delay constraints at the transmitter's queue. We also investigate the impact of the proposed mode selection mechanism of D2D communication on the EC analysis. In addition, we extend our analysis to underlay and overlay modes of D2D communication. 
    \item The throughput of a wireless channel changes rapidly with the change in the channel conditions, and this phenomenon becomes more critical in D2D communication due to its opportunistic channel assignment. Therefore, we extend our analysis to scenarios when CSI is available at the receiver only and when CSI is available at both the transmitter and the receiver. 
    \item We note that the packet drop ratio increases quite rapidly when the transmitter sends data without knowing the channel conditions before the transmission. This high packet drop ratio leads to poor QoS performance of the D2D channel. Therefore, we propose using the HARQ retransmission scheme to enhance the QoS performance of RIS-assisted D2D communication. To this end, we provide a complete framework for the integration of HARQ with RIS-assisted D2D communication. Additionally, we perform the throughput analysis of the HARQ-enabled RIS-assisted D2D communication under statistical QoS guarantees to investigate its efficacy. 
\end{itemize}

The remainder of the paper is organized as follows. Section II introduces the system model of the RIS-assisted D2D communication in a single-tier cellular network. Sections III presents the mode selection mechanism. Section IV and Section V provide the statistical QoS analysis of the RIS-assisted D2D communication without and with CSI at the transmitter, respectively. Section VI presents the QoS provisioning of the RIS-assisted D2D link using the HARQ retransmission scheme. Simulation results are presented in Section VII. Finally, Section VIII concludes the paper. 

\section{System Model and Preliminaries}
\subsection{System Model}
\begin{figure}[ht]
\begin{center}
	\includegraphics[width=4in]{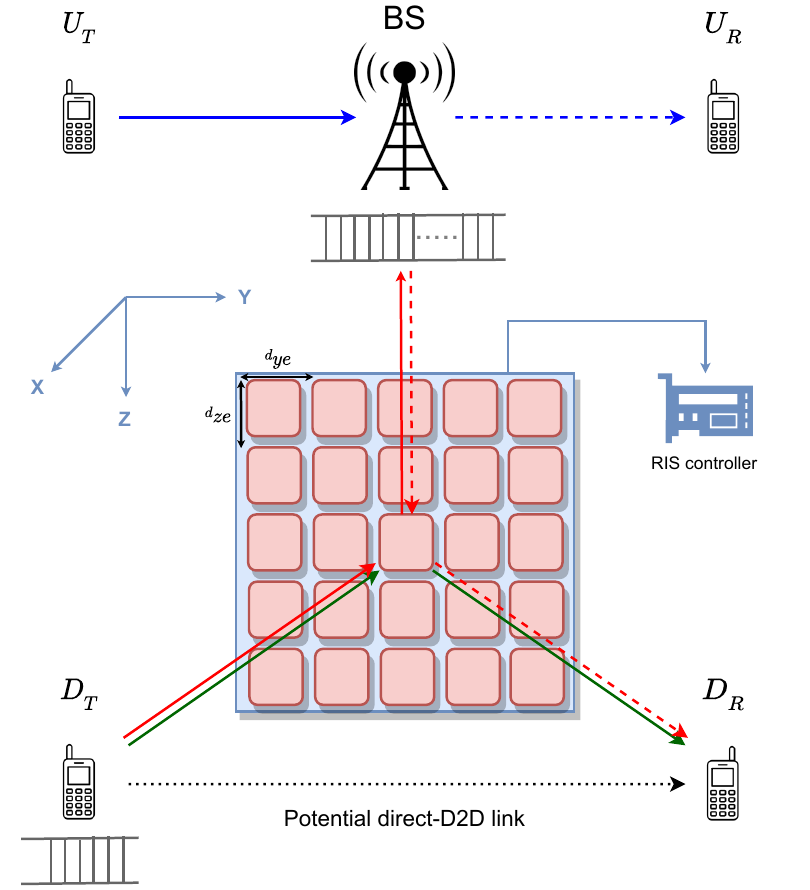}
\caption{System model of RIS-assisted D2D communication: solid green arrows represent RIS-assisted D2D link; solid and striped red arrows are the uplink and downlink of RIS-assisted cellular link, respectively; solid and striped blue arrows are the uplink and downlink of cellular user link; dotted black arrow shows the potential direct-D2D link.}
\label{system_model}
\end{center}
\end{figure}

We consider a single-cell RIS-assisted cellular network with two pairs of user equipment (UE), as shown in Fig. \ref{system_model}. We term one pair as a potential D2D pair, where the transmitter UE ($D_T$ can communicate with the receiver UE ($D_R$) in two modes. One mode is the RIS-assisted D2D mode. The second is the RIS-assisted cellular mode. In the RIS-assisted D2D mode, $D_T$ can communicate with $D_R$ by relaying its data from the RIS elements placed in the network. In the RIS-assisted cellular mode, the RIS elements are used to establish the uplink communication between $D_T$ and BS and the downlink communication between BS and $D_R$. The second pair of UEs ($U_T$ and $U_R$), which is termed cellular user pair, operates only in RIS-assisted cellular mode. For D2D communication, selecting a communication mode at $D_T$ from the available modes can be done by using the mode selection mechanism.

We note that when the D2D pair communicates in the underlay mode, it experiences interference from the cellular users. To this end, RISs can be utilized to effectively eliminate the inter-channel interference caused by the reuse of the channel. It can be achieved by configuring the RIS to control the reflection, refraction, and scattering of electromagnetic waves that impinge on the surface. The signals reflected by the RIS can constructively combine with signals from other paths to boost the desired receiver signal-to-interference-plus-noise ratio (SINR). It is done by using passive beamforming \cite{xu2021ris}. Similarly, when there are multiple D2D users present in a RIS-assisted D2D communication network, interference among multiple transmitted and received signals can be mitigated using the passive beamforming mechanism of the RIS system. However, the impact of user density (multiple D2D pairs in the network) on the QoS performance of the proposed RIS-assisted D2D communication is out of the scope of this work.
\subsection{RIS-assisted D2D Communication}
RIS-assisted D2D communication is a new paradigm of D2D communication in which the throughput of D2D users can be enhanced while reducing the D2D interference by using intelligent and reconfigurable surfaces. These surfaces can control the refraction, reflection, and scattering of the wireless signals and direct them towards the desired receiver. For our analysis, an RIS is a two-dimensional uniform planar array of $N \times N$ elements. Each RIS element introduces a phase shift to direct/reflect the incoming signal from the transmitter ($D_T$, $U_T$, or BS) to the desired destination ($D_R$, $U_R$, or BS). The phase shift range of each element is fixed and can be taken with equal intervals between $[0,2\pi]$. We assume $b_q$ is the number of quantization bits; therefore, $2^{b_q}$ patterns of phase shift values are generated. The reflection coefficient of the RIS element in the $m_z^{\text{th}}$ row and $m_y^{\text{th}}$ column is $\zeta_{m_z,m_y} = e^{j\phi_{m_z,m_y}}$, where $\phi_{m_z,m_y} = \frac{2l_{m_z,m_y}\pi}{2^{b_q}-1}$, and $l_{m_z,m_y} = \{0,1,2,\dots,2^{b_q}-1\}$ and $1\leq m_z,m_y\leq N$. We also note that the phase shifts of the RIS elements needs to be selected in a way that maximizes the gain of the respected channels. Therefore, a careful investigation is required for calculating phase shifts of the RIS elements for different network settings. To this end, we provide a step-by-step procedure for calculating the phase shifts in Appendix \ref{prop_phase_shift}.

For our analysis of RIS-assisted D2D communication, we establish a three dimensional Cartesian coordinate system. We place the RIS on the Y-Z plane, with $d_{ye}$ and $d_{yz}$ denoting the spacing between adjacent RIS elements, as shown in Fig. \ref{system_model}. $M_{m_z,m_y} = \{0,m_yd_{ye},m_zd_{ze}\}$ are the locations of the RIS elements which are calculated with reference to the vertex in the bottom right corner of the RIS. We place BS, D2D, and the cellular pair on the X-Y plane (without the negative X-axis). Hence, for the D2D pair, $D_T$ is positioned at $M_{D_T} = \{D_{Tx},D_{Ty},0\}$ and $D_R$ is positioned at $M_{D_R} = \{D_{Rx},D_{Ry},0\}$. Similarly, the cellular pair and BS are positioned at $M_{U_T} = \{U_{Tx},U_{Ty},0\}$, $M_{U_R} = \{U_{Rx},U_{Ry},0\}$, and $M_{BS} = \{BS_{x},BS_{y},0\}$, respectively. Using these coordinates, we can find the distance between $D_T$ and $D_R$ and between $D_T$ and BS, which then can be used for mode selection mechanism. The corresponding distances are as follows:
\begin{subequations}\label{distance}
  \begin{align}
  d^{^{D_T}}_{m_z,m_y} &= \sqrt{(D_{Tx})^2 + (D_{Ty}-m_yd_{ye})^2 + (-m_zd_{ze})^2}\\
  d^{^{m_z,m_y}}_{D_R} &= \sqrt{(D_{Rx})^2 + (D_{Ry}-m_yd_{ye})^2 + (-m_zd_{ze})^2}\\
  d^{^{m_z,m_y}}_{BS} &= \sqrt{(BS_x)^2 + (BS_y-m_yd_{ye})^2 + (-m_zd_{ze})^2},
  \end{align}
\end{subequations}
where $d^{^{D_T}}_{m_z,m_y}$ is the distance between $D_T$ and the center of the RIS, $d^{^{m_z,m_y}}_{D_R}$ is the distance between the center of the RIS and $D_R$, and $d^{^{m_z,m_y}}_{BS}$ is the distance between the center of the RIS and the BS. We also assume that $d^{^{m_z,m_y}}_{BS} = d_{m_z,m_y}^{BS}$. By using (\ref{distance}), we can find the distances for the RIS-assisted D2D mode and RIS-assisted cellular mode, which are $D_d = d^{^{D_T}}_{m_z,m_y} + d^{^{m_z,m_y}}_{D_R}$ and $D_{c,1} = d^{^{D_T}}_{m_z,m_y} + d^{^{m_z,m_y}}_{BS}$, respectively. Moreover, we term the transmission link for both of the communication modes as a virtual line of sight (VLoS) path because the RIS reflected beam is directional. This VLoS path can be considered as the dominant path among various multipaths because an RIS-assisted channel is a reflection-dominated channel. Therefore, we use the Rician distribution model in our analysis where this VLoS path is considered as the dominant LoS path and all other multipaths are considered as non-line of sight (NLoS). Therefore, the channel for the RIS-assisted D2D mode can be written as:
\begin{equation}\label{D2D_channel}
  h^{^{D_T,D_R}}_{m_z,m_y} = \sqrt{\frac{\alpha}{1+\alpha}}\Upsilon_{m_z,m_y}^{^{D_T,D_R}} + \sqrt{\frac{1}{1+\alpha}}\Upsilon_{_{NLoS}}^{^{D_T,D_R}}
\end{equation}
where $\alpha$ is the Rician factor, and $\Upsilon_{m_z,m_y}^{^{D_T,D_R}}$ and $\Upsilon_{_{NLoS}}^{^{D_T,D_R}}$ are the LoS and NLoS components of the channel, respectively. They can be calculated as
\begin{subequations}\label{LoS_and_NLoS_D2D}
  \begin{align}
  \Upsilon_{m_z,m_y}^{^{D_T,D_R}} &= \sqrt{PL(d^{^{D_T}}_{m_z,m_y}  d^{^{m_z,m_y}}_{D_R})} e^{\frac{-j2\pi}{\lambda}D_d}\\
  \Upsilon_{_{NLoS}}^{^{D_T,D_R}} &=\sqrt{PL(d^{^{D_T}}_{m_z,m_y}  d^{^{m_z,m_y}}_{D_R})}\xi^{^{D_T,D_R}}_{_{NLoS}}
  \end{align}
\end{subequations}
where $\lambda$ is the wavelength and $\xi^{^{D_T,D_R}}_{_{NLoS}} \sim \mathcal{CN}(0,1)$ is the small scale fading of the NLoS component of the channel which follows a complex normal distribution with zero mean and unity variance. $PL(d^{^{D_T}}_{m_z,m_y}  d^{^{m_z,m_y}}_{D_R})$ is the path loss of the RIS-assisted D2D channel, which can be written as \cite{9206044}:
\begin{equation}\label{pathloss_direct}
    PL(d^{^{D_T}}_{m_z,m_y}  d^{^{m_z,m_y}}_{D_R}) = PL_d=\frac{64\pi^3(d^{^{D_T}}_{m_z,m_y}  d^{^{m_z,m_y}}_{D_R})^2}{G G_rG_tN^2d_{ye}d_{ze}\lambda^2F(\upsilon_{_{D_T}},\mu_{_{D_T}})F(\upsilon_{_{D_R}},\mu_{_{D_R}})}
\end{equation}
where $G_t$, $G_r$, and $G$ are the gains of $D_T$, $D_R$, and the RIS unit cell, respectively. $F(\upsilon,\mu)$ is the normalized power radiation pattern and $\upsilon_{_{D_T}}$ and $\mu_{_{D_T}}$ ($\upsilon_{_{D_R}}$ and $\mu_{_{D_R}}$) are the elevation and the azimuth angles between $D_T$ and the center of the RIS (between the center of the RIS and $D_R$), respectively. 
As for the RIS-assisted cellular mode, we note that the channel between $D_T$ and $D_R$ is a two-hop channel because transmission is relayed through the BS. Therefore, we have to find two channels in this mode, one from $D_T$ to BS ($D_T \to$RIS$\to$ BS) and the second from BS to $D_R$ (BS$\to$RIS$\to D_R$). These channels are
\begin{subequations}\label{cellular_channel}
  \begin{align}
  h^{^{D_T,BS}}_{m_z,m_y} &= \sqrt{\frac{\alpha}{1+\alpha}}\Upsilon_{m_z,m_y}^{^{D_T,BS}} + \sqrt{\frac{1}{1+\alpha}}\Upsilon_{_{NLoS}}^{^{D_T,BS}}\\
  h^{^{BS,D_R}}_{m_z,m_y} &= \sqrt{\frac{\alpha}{1+\alpha}}\Upsilon_{m_z,m_y}^{^{BS,D_R}} + \sqrt{\frac{1}{1+\alpha}}\Upsilon_{_{NLoS}}^{^{BS,D_R}}
  \end{align}
\end{subequations}
where $\Upsilon_{m_z,m_y}^{^{D_T,BS}}$ and $\Upsilon_{m_z,m_y}^{^{BS,D_R}}$ are the LoS components of the $D_T \to$RIS$\to$ BS and BS$\to$RIS$\to D_R$ channels, respectively. These components can be obtained by finding the path loss of the respective channels. The path loss of the $D_T \to$RIS$\to$ BS and BS$\to$RIS$\to D_R$ channels can be written as \cite{9206044}:
\begin{subequations}\label{pathloss_cellular}
\begin{align}
    PL(d^{^{D_T}}_{m_z,m_y}d^{^{m_z,m_y}}_{\text{BS}}) &= PL_{_{D_T,\text{BS}}} =\frac{64\pi ^3(d^{^{D_T}}_{m_z,m_y}  d^{^{m_z,m_y}}_{\text{BS}})^2}{GG_{\text{BS}}G_tN^2d_{ye}d_{ze}\lambda^2F(\upsilon_{_{D_T}},\mu_{_{D_T}})F(\upsilon_{_{\text{BS}}},\mu_{_{\text{BS}}})},\\
    PL(d^{^{\text{BS}}}_{m_z,m_y}d^{^{m_z,m_y}}_{D_R})&= PL_{_{\text{BS},D_R}}=\frac{64\pi^3(d^{^{\text{BS}}}_{m_z,m_y}d^{^{m_z,m_y}}_{D_R})^2}{GG_rG_{\text{BS}}N^2d_{ye}d_{ze}\lambda^2F(\upsilon_{_{\text{BS}}},\mu_{_{\text{BS}}})F(\upsilon_{_{D_R}},\mu_{_{D_R}})},
\end{align}
\end{subequations}
where $G_{\text{BS}}$ is the gain of the BS and $\upsilon_{_{\text{BS}}}$ and $\mu_{_{\text{BS}}}$ are the elevation and the azimuth angles between the BS and the center of the RIS, respectively. Then, we can find the LoS components of the $D_T \to$RIS$\to$ BS and BS$\to$RIS$\to D_R$ channels ($\Upsilon_{m_z,m_y}^{^{D_T,BS}}$ and $\Upsilon_{m_z,m_y}^{^{BS,D_R}}$) by substituting (\ref{pathloss_cellular}a), (\ref{pathloss_cellular}b), $D_{c,1}$, and $D_{c,2}$ ($= d^{^{\text{BS}}}_{m_z,m_y} + d^{^{m_z,m_y}}_{_{D_R}}$) in (\ref{LoS_and_NLoS_D2D}a). To find the NLoS components of the uplink ($\Upsilon_{_{NLoS}}^{^{D_T,BS}}$) and the downlink channels ($\Upsilon_{_{NLoS}}^{^{BS,D_R}}$), one has to find the small scale fading for both the channels. Small scale fading for the uplink channel can be written as $\xi^{^{D_T,\text{BS}}}_{_{NLoS}} \sim \mathcal{CN}(0,1)$, which follows a complex  normal  distribution  with  zero  mean and unity variance. Similarly, small scale fading for the downlink channel is $\xi^{^{\text{BS},D_R}}_{_{NLoS}} \sim \mathcal{CN}(0,1)$. By substituting these along with (\ref{pathloss_cellular}a) and (\ref{pathloss_cellular}b) in (\ref{LoS_and_NLoS_D2D}b) we can find the NLoS components of the uplink and downlink channels. We note that the transmission channel of the RIS-assisted  cellular mode is a two-hop channel, therefore, it consumes two time slots for the transmission. Further discussion on the rate of the two-hop channel is available in Section IV-A.

\section{Mode Selection} \label{mode_selection}
The proposed mode selection at $D_T$ chooses the best communication mode among the available modes. For the considered system model, it implies selection between RIS-assisted D2D mode ($D_T$ $\to$ RIS $\to$ $D_R$) and RIS-assisted cellular mode ($D_T$ $\to$ RIS $\to$ BS $\to$ RIS $\to$ $D_R$). This selection of a suitable communication mode is based upon the channel quality of the candidate transmission links. There are multiple ways to measure the channel quality of a transmission link, such as received signal strength, signal-to-noise ratio (in case of overlay communication mode), signal-to-interference-and-noise-ratio (in case of underlay communication mode), instantaneous channel state information (CSI), statistical CSI (pathloss), etc. In the proposed mode selection mechanism, we use statistical CSI (pathloss) as the only feature for selecting a communication mode. It is because the pathloss varies slowly in a wireless channel (large scale fading/pathloss typically changes over the interval of seconds and minutes), and once estimated, the observation/measurement of pathloss remains valid for multiple seconds. Whereas, the instantaneous CSI varies fast due to the small-scale fading (even for a stationary wireless channel, small-scale fading should be estimated several times in a time slot), which makes it a less desirable feature (due to heavy signaling overhead) for the mode selection mechanism.

The mode selection for the RIS-assisted D2D communication for uplink transmission is mapped to the following binary hypothesis testing (BHT) problem:
\begin{equation}
	\label{eq:H0H1}
	 \begin{cases} H_0: & \text{RIS-assisted D2D mode ($D_T$ $\to$ RIS $\to$ $D_R$)} \\
       H_1: & \text{RIS-assisted cellular mode \bigg(\begin{tabular}[c]{@{}l@{}}$D_T \to$ RIS $\to$ BS\\ $\to$ RIS $\to D_R$\end{tabular}\bigg).} \end{cases}
\end{equation}

In this mode selection method, we use the path loss of the respective links as a feature to construct the BHT. When the path loss of the RIS-assisted D2D link is less than the path loss of the RIS-assisted cellular uplink, $D_T$ transmits data on the RIS-assisted D2D link. Otherwise, $D_T$ uses the RIS-assisted cellular uplink for data transmission. Let $\widehat{PL}_{d}$ and $\widehat{PL}_{_{D_T,\text{BS}}}$ denote the noisy measurement of $PL_d$ and $PL_{_{D_T,\text{BS}}}$, respectively. $PL_d$ and $PL_{_{D_T,\text{BS}}}$ are the path loss of the RIS-assisted D2D link ($D_T \to$ RIS $\to D_R$) and RIS-assisted cellular uplink ($D_T \to$ RIS $\to$ BS), respectively, and can be found using \eqref{pathloss_direct} and \eqref{pathloss_cellular}, respectively. We assume that $\widehat{PL}$ follows a Gaussian distribution with $PL$ as the mean and $\sigma_{_{PL}}^2$ as the variance of the estimation error of the path loss measurement. Thus, we can write $\widehat{PL} \sim \mathcal{N} (PL,\sigma^2_{_{PL}})$. Note that $D_T$ uses its immediate channels, which are RIS-assisted D2D channel and RIS-assisted cellular uplink, for the mode selection. It is because $D_T$ does not have global information of the network, and the BS is responsible for the downlink transmission to $D_R$ in the RIS-assisted cellular mode. Therefore, we establish the test statistic $\tau = \widehat{PL}_d - \widehat{PL}_{_{D_T,\text{BS}}}$, which is inline with the prior work \cite{mahmood2013mode}. By using $\tau$, the probability distributions of the RIS-assisted cellular and RIS-assisted D2D modes become $\tau|H_1\sim \mathcal{N}(m_{\tau},\sigma_{\tau}^2)$ and $\tau|H_0\sim \mathcal{N}(-m_{\tau},\sigma_{\tau}^2)$, respectively. Where, the mean $m_{\tau}$ can be calculated by taking the difference of the path loss of the communication modes ($m_{\tau} = PL_{d}-PL_{_{D_T,\text{BS}}}$), and the variance $\sigma_{\tau}^2$ is equal to twice the variance of the estimation error of the path loss measurement ($2\sigma_{_{PL}}^2$). Without loss of generality, let $m_{\tau} > 0$. By implementing log-likelihood ratio test, the mode selection problem can be written as,
\begin{equation}\label{LLRT}
\log_e (\frac{p(\tau|H_1)}{p(\tau|H_0)}) \underset{H_0}{\overset{H_1}{\gtrless}} \log_e \frac{\pi_0}{\pi_1} \implies \tau \underset{H_0}{\overset{H_1}{\gtrless}} \log_e (\frac{\pi_0}{\pi_1})\frac{\sigma_{\tau}^2}{2m_{\tau}}.
\end{equation}
where $\pi_0$ and $\pi_1$ are the prior probabilities of RIS-assisted D2D and RIS-assisted cellular modes, respectively. 
%Note that we assume only the BS is capable of measuring the global/cell-level information (i.e., pathloss measurements of each communication link). Therefore, the mode selection is carried out at the BS, and the outcome is broadcast to the D2D pair ($D_T$ and $D_R$) on the downlink control channel. This assumption is inline with the prior work [4].

To measure the performance of the proposed BHT problem, we use the error and correct detection probabilities. The type-1 error probability can be written as:
\begin{equation} \label{P_e1}
\begin{split}
P_{e,1} = P(H_1|H_0) &= P\bigg(\tau > \log_e (\frac{\pi_0}{\pi_1})\frac{\sigma_{\tau}^2}{2m_{\tau}}|H_0\bigg)= Q\bigg(\frac{\log_e (\pi_0/\pi_1)\sigma_{\tau}^2 + 2m_{\tau}^2}{2m_{\tau}\sigma_{\tau}} \bigg)
\end{split}
\end{equation}
where $Q(x)$ is the Marcum Q-function and can be written as $Q(x)=\frac{1}{\sqrt{2\pi}} \int_x^\infty  e^{-\frac{t^2}{2}} dt$. Similarly, the type-2 error probability becomes
\begin{equation}\label{P_e2}
\begin{split}
P_{e,2} = P(H_0|H_1) &= P\bigg(\tau < \log_e (\frac{\pi_0}{\pi_1})\frac{\sigma_{\tau}^2}{2m_{\tau}}|H_1\bigg)=1- Q\bigg(\frac{\log_e (\pi_0/\pi_1)\sigma_{\tau}^2 - 2m_{\tau}^2}{2m_{\tau}\sigma_{\tau}} \bigg).
\end{split}
\end{equation}
Now, the probability of correct detection of $H_0$ and $H_1$ can be written as the following: $P_{d,1} = P(H_0|H_0) = 1-P_{e,1}$ and $P_{d,2} = P(H_1|H_1) = 1-P_{e,2}$. Note that for equal priors ($\pi_0 = \pi_1$), $\log_e (\pi_0/\pi_1) = 0$. Consequently, the error probabilities in (\ref{P_e1}) and (\ref{P_e2}) would become $P_{e,1} = Q(\frac{m_{\tau}}{\sigma_{\tau}})$ and $P_{e,2} = 1-Q(\frac{-m_{\tau}}{\sigma_{\tau}})$, respectively. These error and correct detection probabilities will be used to find the impact of mode selection on the EC analysis of D2D communication (see Section IV and V). Additionally, to measure the reliability of the BHT problem and the path loss measurement, we use the Kullback-Leibler divergence (KLD). In our case, the KLD $D(p(\tau|H_1)||p(\tau|H_0))$ is given as \footnote{In general, Jensen-Shannon divergence (JSD) should be preferred over KLD because JSD is a true distance measure, while KLD is not. However, in our case, we have enough symmetry such that $D(p(\tau|H_1)||p(\tau|H_0)) = D(p(\tau|H_0)||p(\tau|H_1))$, so it suffices to use KLD. }: $D = \int_{-\infty}^{\infty} p(\tau|H_1) \log(\frac{p(\tau|H_1)}{p(T|H_0)}) d\tau = \frac{m_{\tau}^2}{\sigma_{\tau}^2}$.

Next, we perform statistical QoS analysis for RIS-assisted D2D communication in cellular networks. Note that data transmission on a wireless channel depends on the varying channel conditions. If the transmitter knows the CSI prior to the transmission, it uses this information to adjust its transmission rate and transmit power accordingly. On the other hand, when the transmitter does not know the CSI prior to the transmission, it transmits data using a fixed transmission rate and transmit power. In our analysis, we first assume that the transmitter is unaware of the channel conditions and performs QoS analysis using a fixed rate and a fixed power. Later, we also assume that the channel conditions are known at the transmitter and then the transmitter performs QoS analysis using variable transmission rate. %To perform QoS analysis, we use a well-known analytical tool, the Effective Capacity (EC). EC is a link-layer model that can provide maximum throughput of a communication link under QoS constraints, which are imposed at the transmitter's finite-sized queue.
\section{Statistical QoS Analysis With CSI at the Receiver Only }
When $D_T$ has no knowledge of CSI prior to transmission, data is sent at a fixed transmission rate and fixed transmission power. Although transmission without the knowledge of the channel conditions may lead to a higher packet drop ratio, this issue can be addressed by employing retransmission schemes. We discuss the use of retransmission schemes for RIS-assisted D2D communication later (in Section \ref{HARQ}). Transmission without CSI at the transmitter, combined with the mode selection mechanism (see Section \ref{mode_selection}), makes the D2D link a Markovian service process. Let $r_{_{t}}$ be a constant transmission rate of $D_T$; then, there emerges four Markov states, which can be expressed as shown in Table \ref{states}.

\begin{table}[b]
\centering
\caption{Markov chain representation of states without CSIT}
\label{states}
\begin{tabular}{|c|c|c|c|}
\hline
\textbf{States} & \textbf{Description}                                                                             & \textbf{Notation}                                                            & \textbf{Action}                                                  \\ \hline
$s_1$            & \begin{tabular}[c]{@{}c@{}}RIS-assisted D2D mode is\\ selected and link is ON.\end{tabular}       & \begin{tabular}[c]{@{}c@{}}$H_0$ and $r_{_{t}} < C_{d}^{^{\text{RIS}}}(n)$\end{tabular}    & \begin{tabular}[c]{@{}c@{}}Tx with $r_{_{t}}$\end{tabular} \\ \hline
$s_2$           & \begin{tabular}[c]{@{}c@{}}RIS-assisted D2D mode is\\ selected and link is OFF.\end{tabular}      & \begin{tabular}[c]{@{}c@{}}$H_0$ and $r_{_{t}} > C_{d}^{^{\text{RIS}}}(n)$\end{tabular}    & \begin{tabular}[c]{@{}c@{}}No Tx\end{tabular} \\ \hline
$s_3$            & \begin{tabular}[c]{@{}c@{}}RIS-assisted Cellular mode is\\ selected and link is ON.\end{tabular}  & \begin{tabular}[c]{@{}c@{}}$H_1$ and $r_{_{t}} < C_{c}^{^{\text{RIS}}}(n)$\end{tabular}    & \begin{tabular}[c]{@{}c@{}}Tx with $r_{_{t}}$ \end{tabular} \\ \hline
$s_4$            & \begin{tabular}[c]{@{}c@{}}RIS-assisted Cellular mode is\\ selected and link is OFF.\end{tabular} & \begin{tabular}[c]{@{}c@{}}$H_1$ and $r_{_{t}} > C_{c}^{^{\text{RIS}}}(n)$\end{tabular}    & \begin{tabular}[c]{@{}c@{}}No Tx\end{tabular} \\ \hline
\end{tabular}
\end{table}
$C_{d}^{^{\text{RIS}}}(n)$ and $C_{c}^{^{\text{RIS}}}(n)$ are the channel capacities of the RIS-assisted D2D and RIS-assisted cellular links, respectively, for time slot $n$. When $r>C_{d}^{^{\text{RIS}}}(n)$ or $r>C_{c}^{^{\text{RIS}}}(n)$ there emerges a backlog of data at the transmitter's queue. Due to the finite-sized nature of the queue, this data backlog will choke the queue at steady-state position. Therefore, we consider a state as being in the OFF condition when the transmission rate is larger than the instantaneous channel capacity of the respected link. Thus, $s_1$ and $s_3$ states are considered to be in the ON state, and $s_2$ and $s_4$ are considered to be in the OFF state. In short, transmission is only possible in two states and in remaining two states $D_T$ does not transmit data, as shown in Fig. \ref{queue}. Now, to find the state transition probabilities, we have to find $C_{d}^{^{\text{RIS}}}(n)$ and $C_{c}^{^{\text{RIS}}}(n)$. Additionally, we note that D2D communication in cellular networks can be done either in overlay settings (with orthogonal channel assignment to cellular and D2D users) or in underlay settings (non-orthogonal channel assignment to cellular and D2D users). Therefore, we will find $C_{d}^{^{\text{RIS}}}(n)$ and $C_{c}^{^{\text{RIS}}}(n)$ for both the settings.
\begin{figure}[ht]
\begin{center}
	\includegraphics[width=4in]{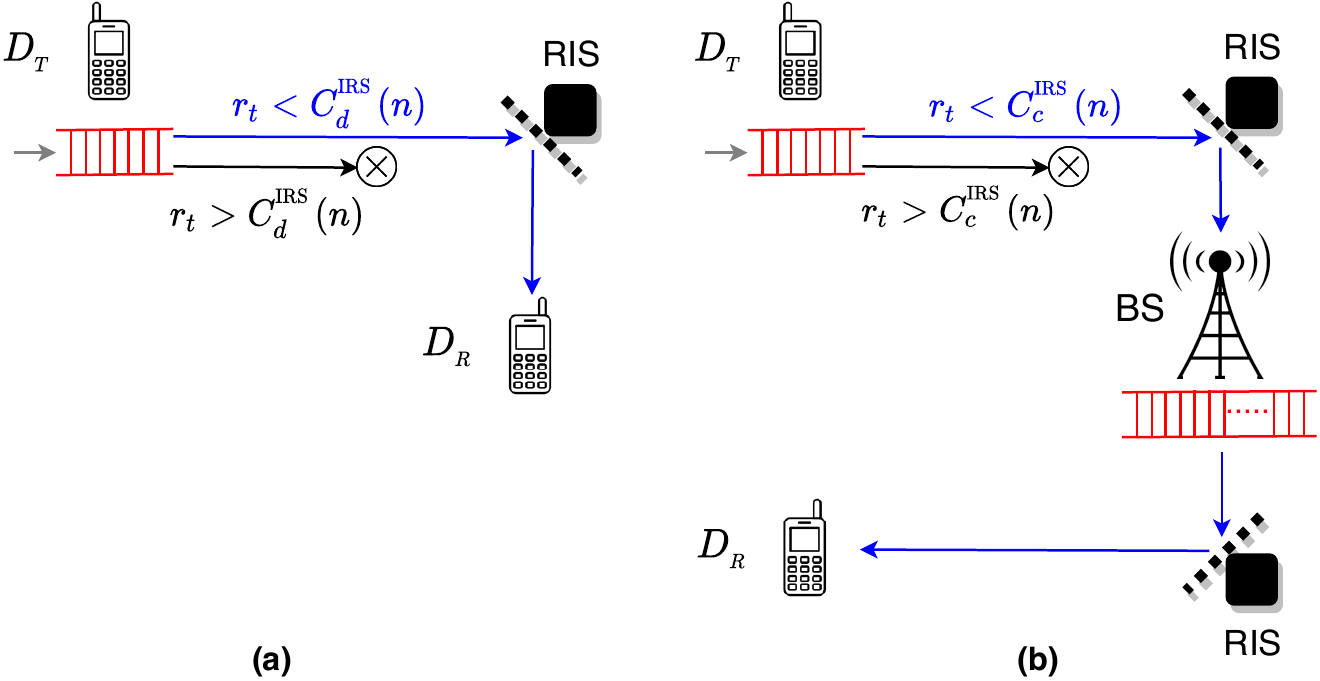}
\caption{Transmission with CSI available at $D_R$ only: (a) RIS-assisted D2D mode; (b) RIS-assisted cellular mode.}
\label{queue}
\end{center}
\end{figure}
\subsection{EC of RIS-assisted Underlay-D2D}
In underlay D2D settings, $D_T$ transmits data by reusing cellular user's resources. Therefore, transmitted signal experience interference from $U_T$. In this case, the instantaneous channel capacity of $D_T \to$ RIS $\to D_R$ link can be written as,
\begin{equation}\label{C_d_underlay}
  C_{d,u}^{^{\text{RIS}}}(n) = B\log_2(1+ \Gamma_d^{^{\text{RIS}}}(n))
\end{equation}
where $\Gamma_d^{^{\text{RIS}}}(n)$ is the SINR of the $D_T \to$ RIS $\to D_R$ link. This SINR can be written as,
\begin{equation}\label{sinr_d}
  \Gamma_d^{^{\text{RIS}}}(n) = \frac{\bar{P}_{_{D_T}}\big(\sum_{m_z,m_y}h^{^{D_T,D_R}}_{m_z,m_y}(n)e^{j\phi_{m_z,m_y}}\big)^2\big/PL_d}{\bar{P}_{_{U_T}}\big(h_{_{U_T,D_R}})^2\big/PL_{_{U_T,D_R}}+\omega_{_{0}}}
\end{equation}
where $\bar{P}_{_{D_T}}$ and $\bar{P}_{_{U_T}}$ are the average transmit powers of $D_T$ and $U_T$, respectively. $PL_d$ ($PL_{_{U_T,D_R}}$) and $h^{^{D_T,D_R}}_{m_z,m_y}$ ($h_{_{U_T,D_R}}$) are the path loss and channel coefficients of $D_T \to$ RIS $\to D_R$ ($U_T \to D_R$) link, respectively. $\omega_{_{0}}$ is the noise variance of the Additive White Gaussian Noise.

Similarly, for the case of RIS-assisted cellular link, the instantaneous channel capacity can be written as,
\begin{equation}\label{c_c_underlay}
  C_{c,u}^{^{\text{RIS}}}(n) = 0.5\min\{C_{ul,u}^{^{\text{RIS}}}(n), C_{dl,u}^{^{\text{RIS}}}(n)\}
\end{equation}
where $C_{ul,u}^{^{\text{RIS}}}(n)$ and $C_{dl,u}^{^{\text{RIS}}}(n)$ are the instantaneous channel capacities of the uplink ($D_T \to$ RIS $\to$ BS) and the downlink (BS $\to$ RIS $\to D_R$) channels, respectively. The constant 0.5 is due to the fact that the transmission in RIS-assisted cellular link is a two-hop communication and consumes two time slots. Individual channel capacities of the uplink and downlink can be written as,
\begin{subequations}\label{c_ul_dl_underlay}
  \begin{align}
  C_{ul,u}^{^{\text{RIS}}}(n) &= B\log_2(1+\Gamma_{ul}^{^{\text{RIS}}}(n))\\
  C_{dl,u}^{^{\text{RIS}}}(n) &= B\log_2(1+\Gamma_{dl}^{^{\text{RIS}}}(n))
  \end{align}
\end{subequations}
where $\Gamma_{ul}^{^{\text{RIS}}}(n)$ and $\Gamma_{dl}^{^{\text{RIS}}}(n)$ are the SINRs of the uplink and the downlink channels, respectively. These SINRs can be written as follows:
\begin{equation}\label{sinr_ul_dl}
\begin{split}
\Gamma_{ul}^{^{\text{RIS}}}(n)&=\frac{\bar{P}_{_{D_T}}\big(\sum_{m_z,m_y}h^{^{D_T,\text{BS}}}_{m_z,m_y}(n)e^{j\phi_{m_z,m_y}}\big)^2\big/PL_{_{D_T,\text{BS}}}}{\bar{P}_{_{U_T}}\big(h_{_{U_T,\text{BS}}})^2\big/PL_{_{U_T,\text{BS}}}+\omega_{_{0}}}\\
\Gamma_{dl}^{^{\text{RIS}}}(n)&=\frac{\bar{P}_{_{\text{BS}}}\big(\sum_{m_z,m_y}h^{^{\text{BS},D_R}}_{m_z,m_y}(n)e^{j\phi_{m_z,m_y}}\big)^2\big/PL_{_{\text{BS},D_R}}}{\bar{P}_{_{U_T}}\big(h_{_{U_T,D_R}})^2\big/PL_{_{U_T,D_R}}+\omega_{_{0}}}
\end{split}
\end{equation}
where $\bar{P}_{_{\text{BS}}}$ is the average transmit power of the BS. $PL_{_{D_T,\text{BS}}}$ ($PL_{_{\text{BS},D_R}}$) and $h^{^{D_T,\text{BS}}}_{m_z,m_y}$ ($h^{^{\text{BS},D_R}}_{m_z,m_y}$) are the path loss and channel coefficients of the uplink (downlink) channel, respectively. $PL_{_{U_T,\text{BS}}}$ and $h_{_{U_T,\text{BS}}}$ are the path loss and channel coefficients of the $U_T \to$ BS link (uplink interference link). Now, by substituting results from \eqref{c_ul_dl_underlay} and \eqref{sinr_ul_dl} into \eqref{c_c_underlay}, the instantaneous channel capacity of the RIS-assisted cellular link becomes the following:
\begin{equation}\label{c_c_underlay_final}
  C_{c,u}^{^{\text{RIS}}}(n) = 0.5B\log_2(1+\Gamma_c^{^{\text{RIS}}}(n))
\end{equation}
where $\Gamma_c^{^{\text{RIS}}}(n) = \min\{\Gamma_{ul}^{^{\text{RIS}}}(n), \Gamma_{dl}^{^{\text{RIS}}}(n)\}$.

Next, we find the state transition probabilities for states $s_1$, $s_2$, $s_3$, and $s_4$, given in Table \ref{states}. Note that due to the block fading nature of the D2D channel, state change only happens after a block length ($T_b$). Let $p_{c,d}$ be the transition probability from state $c$ to state $d$; then, the transition probability from $s_1$ to $s_2$ in underlay settings can be written as:
\begin{subequations}\label{transition_prob1}
\begin{align}
  p_{1,1}^{u} = P[\tau|H_0&(n)<\delta \enspace \text{and} \enspace r_t<C_{d,u}^{^{\text{RIS}}}(n)\big|\tau|H_0(n-1)<\delta \enspace \text{and} \enspace r_t<C_{d,u}^{^{\text{RIS}}}(n-1)] \\
    \labelrel={a} P[\tau|H_0&(n)<\delta \enspace \text{and} \enspace \Gamma_d^{^{\text{RIS}}}(n)>\gamma_{_{T}}\big|\tau|H_0(n-1)<\delta \enspace \text{and} \enspace \Gamma_d^{^{\text{RIS}}}(n-1)>\gamma_{_{T}}]\\
  \labelrel={b} P[\tau|H_0&(n)<\delta\big|\tau|H_0(n-1)<\delta]P[\Gamma_d^{^{\text{RIS}}}(n)>\gamma_{_{T}}\big|\Gamma_d^{^{\text{RIS}}}(n-1)>\gamma_{_{T}}]\\
  \labelrel={c} P[\tau|H_0&(n)<\delta]\enspace P[\Gamma_d^{^{\text{RIS}}}(n)>\gamma_{_{T}}]\\
  \labelrel={d} P[\tau|H_0&(n)<\delta]\enspace P[\Psi_d^{^{\text{RIS}}}(n)>\gamma_{_{T}}].
\end{align}
\end{subequations}
where $\delta = \log_e (\pi_0/\pi_1)\sigma_{\tau}^2/2m_{\tau}$, $\gamma_{_{T}} = 2^{r_t/B}-1$ is the required SINR (threshold SINR) for the transmission, and $\Psi_d^{^{\text{RIS}}}(n)$ is the signal-to-interference-ratio (SIR) of the RIS-assisted D2D link. Equality \eqref{a} follows from the fact that the condition on the transmission rate can be translated into a condition on the SINR of the RIS-assisted D2D link. Equality \eqref{b} holds because the mode selection mechanism and the fading process are independent of each other. Equality \eqref{c} is due to the fact that the mode selection mechanism, as well as the fading process, change independently among time slots and thus, are memoryless stochastic processes. Moreover, we consider an interference-limited scenario in which we neglect the effects of noise on the transmitted signal to find the SIR of the RIS-assisted D2D link. Equality \eqref{d} follows from this assumption. Next, we know that $P[\tau|H_0(n)<\delta] = P[H_0|H_0]\pi_0 + P[H_0|H_1]\pi_1$. For $P[\Psi_d^{^{\text{RIS}}}(n)<\gamma_{_{T}}]$, we have the following proposition \ref{prop_outage_d}.
\begin{prop}\label{prop_outage_d}
  The outage probability of the SIR of RIS-assisted D2D link in underlay settings is,
  \begin{equation*}
    P[\Psi_d^{^{\text{RIS}}}(n)<\gamma_{_{T}}] = \frac{\bar{P}_{_{U_T}}PL_{d}\gamma_{_{T}}}{PL_{d}\bar{P}_{_{U_T}} + PL_{_{U_T,D_R}}\bar{P}_{_{D_T}}N\pi}.
  \end{equation*}
\end{prop}
\begin{proof}
  Given in Appendix \ref{prop1}.
\end{proof}
Note further that the state transition probability does not depend on the original state. Therefore, $p_{1,1}^{u} = p_1^{u}$ becomes the following
\begin{equation}\label{p_1}
  p_{1,1}^{u} = p_1^{u} = \bigg(P[H_0|H_0]\pi_0 + P[H_0|H_1]\pi_1\bigg) \enspace P[\Psi_d^{^{\text{RIS}}}(n)>\gamma_{_{T}}].
\end{equation}
Similarly, state transition probabilities for states $s_2$, $s_3$, and $s_4$ becomes:
\begin{equation}\label{p_2_3_4}
  \begin{split}
  p_{c,2}^{u} = p_2^{u} &= \bigg(P[H_0|H_0]\pi_0 + P[H_0|H_1]\pi_1\bigg) \enspace P[\Psi_d^{^{\text{RIS}}}(n)<\gamma_{_{T}}]\\
  p_{c,3}^{u} = p_3^{u} &= \bigg(P[H_1|H_0]\pi_0 + P[H_1|H_1]\pi_1\bigg) \enspace P[\Psi_c^{^{\text{RIS}}}(n)>\gamma_{_{T}}]\\
  p_{c,4}^{u} = p_4^{u} &= \bigg(P[H_1|H_0]\pi_0 + P[H_1|H_1]\pi_1\bigg) \enspace P[\Psi_c^{^{\text{RIS}}}(n)<\gamma_{_{T}}]
  \end{split}
\end{equation}
where $\Psi_c^{^{\text{RIS}}}(n)$ is the SIR of the RIS-assisted cellular link, and $P[\Psi_d^{^{\text{RIS}}}(n)>\gamma_{_{T}}] = 1-P[\Psi_d^{^{\text{RIS}}}(n)<\gamma_{_{T}}]$. To find $P[\Psi_c^{^{\text{RIS}}}(n)<\gamma_{_{T}}]$, we have proposition \ref{prop_outage_c}. Similarly, $P[\Psi_c^{^{\text{RIS}}}(n)>\gamma_{_{T}}] = 1-P[\Psi_c^{^{\text{RIS}}}(n)<\gamma_{_{T}}]$.

\begin{prop}\label{prop_outage_c}
  The outage probability of the SIR of the RIS-assisted cellular link in underlay settings is,
\begin{equation*}
\begin{split}
&P[\Psi_c^{^{\text{RIS}}}(n)<\gamma_{_{T}}] =\frac{\bar{P}_{_{U_T}}\gamma_{_{T}}\big[PL_{_{D_T,\text{BS}}}\big\{\bar{P}_{_{U_T}}PL_{_{\text{BS},D_R}}(2-\gamma_{_{T}})+\Omega_1\big\}-PL_{_{\text{BS},D_R}}\Omega_2\big]}{(\bar{P}_{_{U_T}}PL_{_{D_T,\text{BS}}} + \Omega_1)(\bar{P}_{_{U_T}}PL_{_{\text{BS},D_R}} +\Omega_2)}.    
\end{split}
\end{equation*}
Where $\Omega_1 =  \bar{P}_{_{D_T}}PL_{_{U_T,\text{BS}}}N\pi$ and $\Omega_2 =  \bar{P}_{_{\text{BS}}}PL_{_{U_T,D_R}}N\pi$.
\end{prop}
\begin{proof}
  Given in Appendix \ref{prop2}.
\end{proof}
With this, each row of the state transition probability matrix in underlay settings $\mathbf{P}^{u}$ becomes $[p_1^u,p_2^u,p_3^u,p_4^u]$. Because of the identical rows, rank of matrix $\mathbf{P}^u$ becomes 1. Next, we utilize matrix $\mathbf{P}^u$ to find the EC of RIS-assisted D2D communication in underlay settings.

The EC of a point-to-point link was first proposed in \cite{wu2003effective} as:
\begin{equation}\label{EC_main}
  EC(\varphi) = \frac{-\Lambda(-\varphi)}{\varphi} = -\lim_{t\to \inf}\frac{1}{\varphi t}\log_e E [e^{-\varphi S(t)}]
\end{equation}
where $S(t)$ is the cumulative service process and can be defined as $\sum_{n=1}^{t}s(n)$ and where $s(n)$ is the channel service. $\varphi$ is the QoS exponent and can be expressed in terms of delay violation probability at the transmitter's queue. In our analysis, service process is a Markov modulated process; therefore, $\frac{-\Lambda(-\varphi)}{\varphi}$ can be written as $\frac{1}{\varphi}\log_e \text{sp}(\Phi(\varphi)\mathbf{P})$. This result holds for any scenario in which service processes can be regarded as Markov modulated processes \cite{chang2012performance}. Thus, the EC in our analysis can be calculated using the following results:
\begin{equation}\label{EC_sp_main}
  EC_u^{^{\text{RIS}}}(\varphi) = \frac{-\Lambda(-\varphi)}{\varphi} = \frac{1}{\varphi}\log_e \text{sp}(\Phi(\varphi)\mathbf{P}^u)
\end{equation}
where $\text{sp}(\cdot)$ is the spectral radius (largest absolute eigenvalue) of $\Phi(\varphi)\mathbf{P}^u$; $\Phi(\varphi)$ is the diagonal matrix of the log moment generating functions (LMGFs) of the four states given in Table \ref{states}; and $\mathbf{P}^u$ is the state transition probability matrix given above.

From Table \ref{states}, we can see that for state $s_1$ and $s_3$, transmission is possible with $r_t$. Therefore, $s(n) = r_t$ for both of these states; consequently, MGF for both of these states becomes $e^{r_t\varphi}$. On the other hand, for states $s_2$ and $s_4$, the link is OFF and no transmission is possible. Therefore, $s(n) = 0$ for both of these states; consequently, MGF for these states becomes 1. With these MGFs for the four states, $\Phi(\varphi) = \text{diag}[e^{r_t\varphi}, 1, e^{r_t\varphi}, 1]$. Next, to find $\Phi(\varphi)\mathbf{P}$, we note that it is also a unit rank matrix. Therefore, finding trace of $\Phi(\varphi)\mathbf{P}$ is equivalent to finding its spectral radius. Thus, the EC of RIS-assisted D2D communication in underlay settings is:
\begin{equation}\label{EC_overlay_final}
  EC_u^{^{\text{RIS}}}(\varphi) = \frac{-1}{\varphi}[\log_e(p_1^ue^{-r_t\varphi}+ p_2^u+ p_3^ue^{-r_t\varphi}+ p_4^u)].
\end{equation}

From \eqref{EC_overlay_final}, we can also find the optimal transmission rate $r_t^{\text{opt}}$ to maximize $EC_u^{^{\text{RIS}}}(\varphi)$ as follows:
$r_t^{\text{opt}} = \arg \max_{r_t >0}EC_u^{^{\text{RIS}}}(\varphi)$. \footnote{The optimal transmission rate $r_t^{\text{opt}}$ has to be recomputed every time path loss of the RIS-assisted D2D link or the RIS-assisted uplink of the RIS-assisted cellular mode changes. This change can happen either due to the D2D users's mobility or due to change in the RIS elements' configuration.} This maximization problem can be translated into the following minimization problem:
\begin{equation}
    r_t^{\text{opt}} = \arg \min_{r_t >0} \{p_1^ue^{-r_t\varphi}+ p_2^u+ p_3^ue^{-r_t\varphi}+ p_4^u\}.
\end{equation}
Further, we can see from Table \ref{states} that $s_2$ and $s_4$ are OFF states and no transmission is possible during an OFF state. Therefore, probabilities $p^u_2$ and $p^u_4$ do not contribute towards finding the optimal transmission rate. After discarding the irrelevant terms, the final minimization problem for the optimal transmission rate becomes:
\begin{equation}\label{opt_rate}
    r_t^{\text{opt}} = \arg \min_{r_t >0} \{e^{-r_t\varphi}(p_1^u+p_3^u)\}.
\end{equation}
Let $C_f = e^{-r_t\varphi}(p_1^u+p_3^u)$ is the cost function. One can verify that $C_f$ is a convex function \cite{brychkov2012some}. As $C_f$ is differentiable w.r.t $r_t$; therefore, to compute the optimal transmission rate $r_t^{\text{opt}}$, we use iterative gradient decent (GD) method. The control law for the GD method in our case is as follows:
\begin{equation}
    r_t(k+1) = r_t(k)-\beta \nabla\big|_{r_t(k)}
\end{equation}
where $\beta$ is the step size, $k$ is the number of the iteration, and $\nabla$ is the gradient of $C_f$. The gradient can be found by taking the derivative of $C_f$ w.r.t $r_t$, which is as follows: $\nabla = \frac{\partial C_f}{\partial r_t} =-\varphi(p_1^u+p_3^u)e^{-r_t\varphi}$.
\begin{remark}
The RIS-assisted cellular mode transfer data from $D_T$ to $D_R$ using two-hop communication link. This implies two queues in the network; one at $D_T$ and the other at the BS. To this end, we assume that the BS has infinite-sized queue (as shown in Fig. \ref{queue}) and infinite resources (as compared to cellular users); therefore, problem of queue overflow does not occur at the BS. However, in case of multi-tier cellular networks, if the D2D users are present in a micro/pico/femto cell, then the associated BS does not necessarily have infinite resources or infinite-sized queue (buffer). In such case, one has to recompute the EC of RIS-assisted cellular mode using the results in \cite{qiao2012effective}.
\end{remark}
\subsection{EC of RIS-assisted Overlay-D2D}
In overlay-D2D settings, $D_T$ transmit data using dedicated resources, which are orthogonal to the resources allocated to the cellular users; hence, the transmitted signal experience no interference from the cellular users. Therefore, to compute the instantaneous channel capacities in overlay settings ($C_{d,o}^{^{\text{RIS}}}(n)$ and $C_{c,o}^{^{\text{RIS}}}(n)$), we compute the signal-to-noise-ratio (SNR) of the corresponding channels. When the D2D link operates in RIS-assisted D2D mode, the SNR is given as: $\gamma_d^{^{\text{RIS}}}(n) =\bar{P}_{_{D_T}}\big(\sum_{m_z,m_y}h^{^{D_T,D_R}}_{m_z,m_y}(n)e^{j\phi_{m_z,m_y}}\big)^2/PL_{d}\omega_{_{0}}$. Similarly, when D2D link operates in RIS-assisted cellular mode, the SNR of uplink ($D_T \to$ RIS $\to$ BS) and downlink (BS $\to$ RIS $\to D_R$) channels are given as: $\gamma_{ul}^{^{\text{RIS}}}(n) = \bar{P}_{_{D_T}}\big(\sum_{m_z,m_y}h^{^{D_T,\text{BS}}}_{m_z,m_y}(n)e^{j\phi_{m_z,m_y}}\big)^2/PL_{_{D_T,\text{BS}}}\omega_{_{0}}$ and $\gamma_{dl}^{^{\text{RIS}}}(n) =\bar{P}_{_{\text{BS}}}\big(\sum_{m_z,m_y}h^{^{\text{BS},D_R}}_{m_z,m_y}(n)e^{j\phi_{m_z,m_y}}\big)^2/PL_{_{\text{BS},D_R}}\omega_{_{0}}$. Where net SNR of the RIS-assisted cellular link becomes: $\gamma_{c}^{^{\text{RIS}}}(n) = \min \{\gamma_{ul}^{^{\text{RIS}}}(n),\gamma_{dl}^{^{\text{RIS}}}(n) \}$. By substituting these SNRs in \eqref{C_d_underlay} and \eqref{c_c_underlay_final}, we can find the instantaneous channel capacities of the RIS-assisted D2D and RIS-assisted cellular modes in overlay settings, respectively. Note further that, to compute the EC of overlay-D2D scenario, we require to recompute the four probabilities $p_1^{o}$, $p_2^{o}$, $p_3^{o}$, and $p_4^{o}$. We know that the SNR $\gamma_d^{^{\text{RIS}}}(n)$ is exponentially distributed, then the CDF of $\gamma_d^{^{\text{RIS}}}(n)$ becomes: $P[\gamma_d^{^{\text{RIS}}}(n)< \gamma_{_{T}}] = 1-e^{\frac{-\gamma_{_{T}}}{\kappa_d}}$. Where $\kappa_d = \E[\gamma_d^{^{\text{RIS}}}(n)]$ is the mean of $\gamma_d^{^{\text{RIS}}}(n)$. We assume that there is no phase error experienced by the reflected signal (i-e. $\phi_{m_z,m_y}=0$), then $\kappa_d$ can be written as: $\kappa_d = \frac{N\pi\bar{P}_{_{D_T}}}{PL_{d}\omega_{_{0}}}$ (Appendix A \cite{aman2020effective}). Similarly, $P[\gamma_d^{^{\text{RIS}}}(n)> \gamma_{_{T}}] = e^{\frac{-\gamma_{_{T}}}{\kappa_d}}$. Then, $p_{c,1}^{o}$ and $p_{c,2}^{o}$ becomes:
\begin{subequations}\label{p_1_2_overlay}
\begin{align}
p_{c,1}^{o} = p_1^o &= \bigg(P[H_0|H_0]\pi_0 + P[H_0|H_1]\pi_1\bigg)e^{\frac{-\gamma_{_{T}}PL_d\omega_{_{0}}}{N\pi\bar{P}_{_{D_T}}}}\\
p_{c,2}^{o} = p_2^o &=\bigg(P[H_0|H_0]\pi_0 + P[H_0|H_1]\pi_1\bigg)1-e^{\frac{-\gamma_{_{T}}PL_d\omega_{_{0}}}{N\pi\bar{P}_{_{D_T}}}}
\end{align}
\end{subequations}

To find $p_{c,3}^{o}$ and $p_{c,4}^{o}$, we have to find $P[\gamma_c^{^{\text{RIS}}}(n)> \gamma_{_{T}}]$ and $P[\gamma_c^{^{\text{RIS}}}(n)< \gamma_{_{T}}]$, respectively, and for that we have to find $\E[\gamma_c^{^{\text{RIS}}}(n)]$.
\begin{prop}\label{prop_SNR_c}
  Mean of the SNR of the RIS-assisted cellular link ($D_T \to$ RIS $\to$ BS $\to$ RIS $\to D_R$) is,
  \begin{equation}
  \E[\gamma_c^{^{\text{RIS}}}(n)] = \frac{N\pi\bar{P}_{_{D_T}}\bar{P}_{_{\text{BS}}}}{\bar{P}_{_{D_T}}PL_{_{\text{BS},D_R}}+\bar{P}_{_{\text{BS}}}PL_{_{D_T,\text{BS}}}}.
  \end{equation}
\end{prop}
\begin{proof}
  Given in Appendix \ref{prop3}.
\end{proof}
Due to proposition \ref{prop_SNR_c}, $P[\gamma_c^{^{\text{RIS}}}(n)>\gamma_{_{T}}] = 1-P[\gamma_c^{^{\text{RIS}}}(n)<\gamma_{_{T}}]$ becomes: $e^{-\gamma_{_{T}}/\kappa_c}$, where $\kappa_c =N\pi\bar{P}_{_{D_T}}\bar{P}_{_{\text{BS}}}\big/(\bar{P}_{_{D_T}}PL_{_{\text{BS},D_R}}+\bar{P}_{_{\text{BS}}}PL_{_{D_T,\text{BS}}})$. Then, $p_{c,3}^{o}$ and $p_{c,4}^{o}$ becomes:
\begin{subequations}\label{p_3_4_overlay}
\begin{align}
p_{c,3}^{o} = p_3^o &= \bigg(P[H_1|H_0]\pi_0 + P[H_1|H_1]\pi_1\bigg)e^{-\gamma_{_{T}}/\kappa_c}\\
p_{c,4}^{o} = p_4^o &=\bigg(P[H_1|H_0]\pi_0 + P[H_1|H_1]\pi_1\bigg)1-e^{-\gamma_{_{T}}/\kappa_c}
\end{align}
\end{subequations}

Finally, one can find the EC of the RIS-assisted D2D communication in overlay settings by substituting $p_{1}^{o}$, $p_{2}^{o}$, $p_{3}^{o}$ and $p_{4}^{o}$ in \eqref{EC_overlay_final}. Similarly, optimal transmission rate for this case can also be calculated by substituting $p_{1}^{o}$ and $p_{3}^{o}$ in \eqref{opt_rate}.

\section{Statistical QoS Analysis with CSI at the Transmitter and Receiver}
In this section, we assume that perfect CSI is available at the receiver as well as at the transmitter.\footnote{ We note that channel estimation in an RIS-assisted wireless channel is quite different from the conventional wireless channel. It is more challenging because the RIS has passive reflecting elements that cannot process the pilot signals to and from users/BS \cite{pan2020reconfigurable}. In case when RIS consists of a small number of reflecting elements, cascaded CSI is sufficient \cite{9180053}. However, channel training overhead associated with cascaded CSI when the RIS elements are large becomes prohibitively high \cite{lin2020adaptive}. One solution is to adjust the RIS phase shifts based on long-term CSI using location (of the user) and the angular information (of the arrival signal) \cite{zhi2020power}. However, acquiring long-term CSI is not a suitable solution for highly mobile users. In short, due to these challenges, we will use imperfect CSI (average and outdated CSI) in our future link-layer QoS analysis of the RIS-assisted D2D communication.} It means that $D_T$ has instantaneous values of $h_{m_z,m_y}^{^{D_T,D_R}}$ and $h_{m_z,m_y}^{^{D_T,\text{BS}}}$. Due to this assumption, $D_T$ adopts its transmission rate according to the channel conditions. Therefore, the transmission rate is no longer a constant value (as in case of without CSIT) rather a random value, equal to the instantaneous channel capacity of the respective communication link. Moreover, we assume that the average transmit power of $D_T$ is smaller than the average transmit power of the BS ($\bar{P}_{_{D_T}} < \bar{P}_{_{\text{BS}}}$). Let $r_t^d(n)$ and $r_t^c(n)$ be the transmission rates of $D_T$ in RIS-assisted D2D and RIS-assisted cellular modes, respectively. Then, this combined with the mode selection mechanism (given in Section \ref{mode_selection}) lead to a four state Markov modulated process, as shown in Table \ref{states_csit}.

\begin{table}[b]
\centering
\caption{Markov chain representation of states with CSIT}
\label{states_csit}
\begin{tabular}{|c|c|c|c|}
\hline
\textbf{States} & \textbf{Description}                                                                                                                & \textbf{Notation}                                                                           & \textbf{Action}                                                         \\ \hline
$s_1$           & \begin{tabular}[c]{@{}c@{}}RIS-assisted D2D mode is\\selected and detected as\\D2D mode. The link is ON.\end{tabular}           & \begin{tabular}[c]{@{}c@{}}$H_0|H_0$ and  $r_t^d(n) = C_d^{^{\text{RIS}}}(n)$\end{tabular} & \begin{tabular}[c]{@{}c@{}}Tx with $r_t^d(n)$.\end{tabular} \\ \hline
$s_2$           & \begin{tabular}[c]{@{}c@{}}RIS-assisted D2D mode is\\selected and detected as\\cellular mode. The link is OFF.\end{tabular}     & \begin{tabular}[c]{@{}c@{}}$H_1|H_0$ and  $r_t^c(n) > C_d^{^{\text{RIS}}}(n)$\end{tabular} & \begin{tabular}[c]{@{}c@{}}No Tx\end{tabular}             \\ \hline
$s_3$           & \begin{tabular}[c]{@{}c@{}}RIS-assisted cellular mode is\\selected and detected as D2D\\mode. The link is ON.\end{tabular}      & \begin{tabular}[c]{@{}c@{}}$H_0|H_1$ and  $r_t^d(n) < C_c^{^{\text{RIS}}}(n)$\end{tabular} & \begin{tabular}[c]{@{}c@{}}Tx with $r_t^d(n)$.\end{tabular} \\ \hline
$s_4$           & \begin{tabular}[c]{@{}c@{}}RIS-assisted cellular mode is\\selected and detected as\\cellular mode. The link is ON.\end{tabular} & \begin{tabular}[c]{@{}c@{}}$H_1|H_1$ and  $r_t^c(n) = C_c^{^{\text{RIS}}}(n)$\end{tabular} & \begin{tabular}[c]{@{}c@{}}Tx with\\$r_t^c(n)$.\end{tabular} \\ \hline
\end{tabular}
\end{table}

When $D_T$ detects RIS-assisted D2D mode through mode selection mechanism (as in states $s_1$ and $s_3$), the transmission link is considered as ON. This is because, in state $s_1$ $D_T$ transmits data with $r_t^d(n)$ which is equal to the instantaneous channel capacity of RIS-assisted D2D link ($C_d^{^{\text{RIS}}}(n)$), and in state $s_3$ $D_T$ transmits with $r_t^d(n)$ again which in this case is smaller than the instantaneous channel capacity of the RIS-assisted cellular link ($C_c^{^{\text{RIS}}}(n)$). Therefore, reliable communication is possible at $r_t^d(n)$ in both of these states. On the other hand, when $D_T$ detects RIS-assisted cellular mode (as in states $s_2$ and $s_4$), the transmission rate $r_t^c(n)$ is either equal (in state $s_4$) or greater (in state $s_2$) than the instantaneous channel capacity of the link. When $r_t^c(n) = C_c^{^{\text{RIS}}}(n)$, reliable communication is possible at $r_t^c(n)$, and the link is considered as ON. Whereas, when $r_t^c(n) > C_d^{^{\text{RIS}}}(n)$ as in state $s_2$, reliable communication is not possible at $r_t^c(n)$, and the link is considered as OFF. In short, we have three ON states ($s_1$, $s_3$, and $s_4$) and one OFF state ($s_2$) for the case when perfect CSI is available at $D_T$ prior to the transmission as shown in Fig. \ref{queue_CSIT}. 
\begin{figure}[ht]
\begin{center}
	\includegraphics[width=4in]{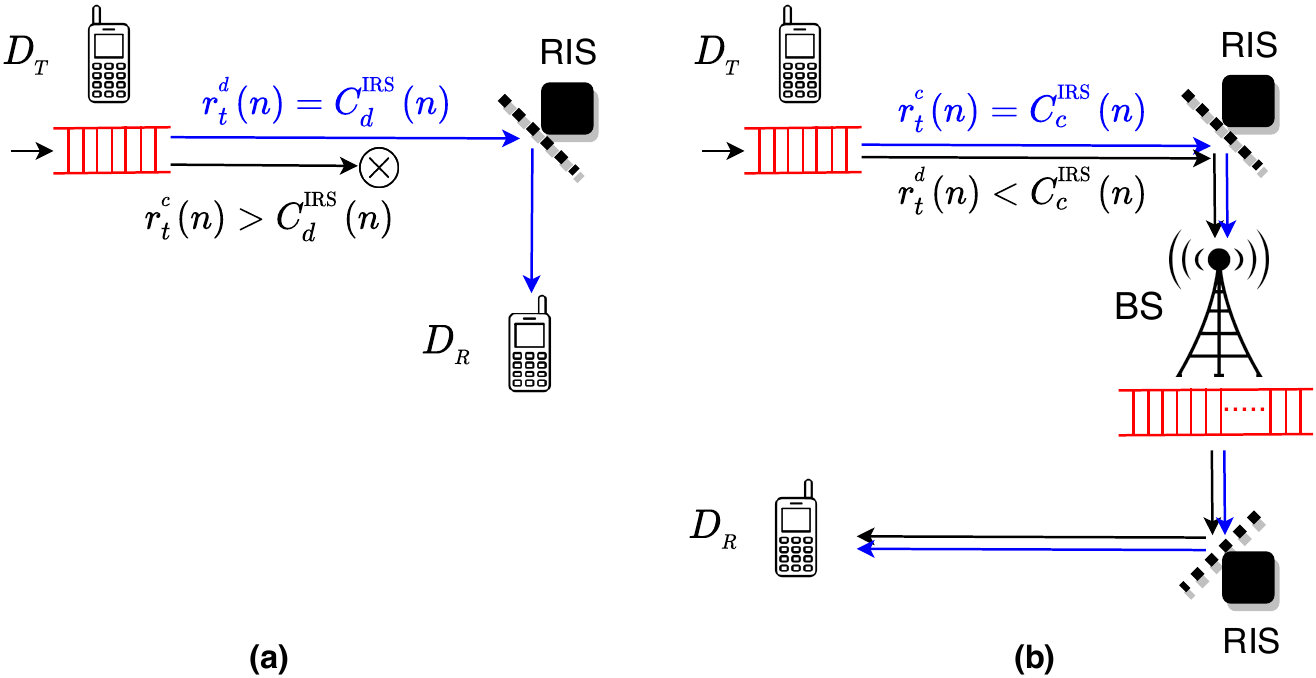}
\caption{Transmission with perfect CSI available at $D_T$ and $D_R$: (a) RIS-assisted D2D mode; (b) RIS-assisted cellular mode.}
\label{queue_CSIT}
\end{center}
\end{figure}

Next, we find the state transition probabilities for states $s_1$, $s_2$, $s_3$, and $s_4$ given in Table \ref{states_csit}. Note that, in this case, transmission rates $r_t^d(n)$ and $r_t^c(n)$ are RVs because of the presence of the perfect CSIT. Due to this fact, the state transition probabilities depend only on the prior and detection probabilities of the RIS-assisted D2D and RIS-assisted cellular modes. This is different from the previous section because when CSIT is not available, the state transition probabilities also depend on the SNR condition of the transmission link along with the detection and the prior probabilities. The state transition probabilities in this case can be written as:
\begin{equation}
\begin{split}
    p_{c,1}=p_1= \pi_0 P_{d,1} &=\pi_0 \bigg[1- Q\bigg(\frac{\log_e (\pi_0/\pi_1)\sigma_{\tau}^2 +2m_{\tau}^2}{2m_{\tau}\sigma_{\tau}} \bigg)\bigg]\\
    p_{c,2}=p_2= \pi_0 P_{e,1} &=\pi_0 \bigg[Q\bigg(\frac{\log_e (\pi_0/\pi_1)\sigma_{\tau}^2 +2m_{\tau}^2}{2m_{\tau}\sigma_{\tau}} \bigg)\bigg]\\
    p_{c,3}=p_3= \pi_1 P_{e,2} &=\pi_1 \bigg[1- Q\bigg(\frac{\log_e (\pi_0/\pi_1)\sigma_{\tau}^2 -2m_{\tau}^2}{2m_{\tau}\sigma_{\tau}} \bigg)\bigg]\\
    p_{c,4}=p_4= \pi_1 P_{d,2} &=\pi_1 \bigg[Q\bigg(\frac{\log_e (\pi_0/\pi_1)\sigma_{\tau}^2 -2m_{\tau}^2}{2m_{\tau}\sigma_{\tau}} \bigg)\bigg].
\end{split}    
\end{equation}

Now, these transition probability values can be used to find the EC for overlay and underlay settings when CSIT is available. Note that, because transition probabilities do not depend on the SNR condition of the transmission link; therefore, these probability values remain same for overlay and underlay settings. 
\subsection{EC of RIS-assisted Underlay-D2D}
To find the EC of RIS-assisted D2D communication in underlay settings when CSIT is available, we have to find $r_{t,u}^d(n)$ and $r_{t,u}^c(n)$. As we know that the transmission rate is equal to the instantaneous channel capacity of the link when CSIT is known; therefore, transmission rates become: $r_{t,u}^d(n) = C_{d,u}^{^{\text{RIS}}}(n) = B\log_2(1+\Gamma_d^{^{\text{RIS}}}(n))$ and $r_{t,u}^c(n) = C_{c,u}^{^{\text{RIS}}}(n) = B\log_2(1+\Gamma_c^{^{\text{RIS}}}(n))$. Where $\Gamma_d^{^{\text{RIS}}}(n)$ is the SINR of the RIS-assisted D2D link and is given in \eqref{sinr_d}. $\Gamma_c^{^{\text{RIS}}}(n)$ is the net SINR of the RIS-assisted cellular links, and can be calculated as $\Gamma_c^{^{\text{RIS}}}(n) = \min\{\Gamma_{ul}^{^{\text{RIS}}}(n),\Gamma_{ul}^{^{\text{RIS}}}(n)\}$. Where $\Gamma_{ul}^{^{\text{RIS}}}(n)$ and $\Gamma_{dl}^{^{\text{RIS}}}(n)$ are the SINR of the uplink and the downlink channels, respectively, and is given in \eqref{sinr_ul_dl}. Let $h_1 = h_{m_z,m_y}^{^{D_T,D_R}}$ and $h_2 = h_{m_z,m_y}^{^{D_T,\text{BS}}} + h_{m_z,m_y}^{^{\text{BS},D_R}}$ be the channel coefficients for the RIS-assisted D2D and the RIS-assisted cellular links, respectively. Then, by using \eqref{EC_overlay_final}, the EC of RIS-assisted D2D communication in underlay settings when CSIT is available becomes,
\begin{equation} \label{EC_with_CSIT_underlay}
    \begin{split}
        EC_u^{^{\text{RIS}}}(\varphi)
        =\frac{-1}{\varphi}\bigg[&\log_e\bigg(\E_{h_1}[e^{-r_{t,u}^d(n)\varphi}](p_1+p_3)+ p_2+\E_{h_2}[e^{-r_{t,u}^c(n)\varphi}p_4]\bigg)\bigg].
    \end{split}
\end{equation}
where $\E_{h_1}$ and $\E_{h_2}$ are the expectations with respect to $h_1$ and $h_2$ channel coefficients, respectively. 

\subsection{EC of RIS-assisted Overlay-D2D}
In this case, transmission rates in RIS-assisted D2D and RIS-assisted cellular modes change because the instantaneous channel capacities in these modes also change. This occurs because in overlay settings, the transmitted signal does not experience interference from cellular users. Hence, the transmission rates for this case become: $r_{t,o}^d(n) = C_{d,o}^{^{\text{RIS}}}(n) = B\log_2(1+\gamma_d^{^{\text{RIS}}}(n))$ and $r_{t,o}^c(n) = C_{c,o}^{^{\text{RIS}}}(n) = B\log_2(1+\gamma_c^{^{\text{RIS}}}(n))$, where $\gamma_d^{^{\text{RIS}}}(n)$ and $\gamma_c^{^{\text{RIS}}}(n)$ are the SNR of the RIS-assisted D2D mode and the net SNR of the RIS-assisted cellular mode, respectively. Then, similar to \eqref{EC_with_CSIT_underlay}, the EC of RIS-assisted D2D communication in overlay settings when CSIT is available becomes,
\begin{equation} \label{EC_with_CSIT_overlay}
    \begin{split}
        EC_o^{^{\text{RIS}}}(\varphi)
        =\frac{-1}{\varphi}\bigg[&\log_e\bigg(\E_{h_1}[e^{-r_{t,o}^d(n)\varphi}](p_1+p_3)+ p_2+\E_{h_2}[e^{-r_{t,o}^c(n)\varphi}p_4]\bigg)\bigg].
    \end{split}
\end{equation}

\begin{remark}{(Potential Direct-D2D Link):}

When a potential direct-D2D link is also available between $D_T$ and $D_R$ in addition to the RIS-assisted D2D link (as shown in Fig. \ref{system_model}), $D_T$ transmits data on both of the links. In such a scenario, the signal received at $D_R$ is the superposition of the signal transmitted on the direct-D2D link and the signal transmitted on the RIS-assisted D2D link. Let $h_{_{D_T,D_R}} = h_{_{D_T}}\sqrt{(d_{_{D_R}}^{^{D_T}})^{-\upsilon}}$ be the channel coefficients of the direct-D2D link, where $d_{_{D_R}}^{^{D_T}}$ is the distance between $D_T$ and $D_R$ ($d_{_{D_R}}^{^{D_T}} = \sqrt{(D_{Tx}-D_{Rx})^2+(D_{Ty}-D_{Ry})^2}$), and $h_{_{D_T}}$ is the small scale fading. Then, the received signal at $D_R$ becomes,
\begin{equation*}
    S_{_{D_R}} = \bar{P}_{_{D_T}}\big(h_{_{D_T,D_R}} + \sum_{m_z,m_y}h^{^{D_T,D_R}}_{m_z,m_y}(n)e^{j\phi_{m_z,m_y}}\big)^2\big/PL_{d}.
\end{equation*}
Further, one can find the EC for this case by recomputing the SINR and SNR expressions for underlay and overlay settings using the above-mentioned expression of the received signal. 
\end{remark}

\section{QoS Provisioning using HARQ with CSI at the Receiver Only} \label{HARQ}
In this section, we explore the retransmission schemes to enhance the QoS performance of the RIS-assisted D2D communication. When a packet is not successfully received by the receiving device, it sends a negative acknowledgment (NACK) to the transmitting device using a secure feedback link. The transmitting device then retransmits the same packet using a retransmission scheme. This phenomenon occurs more often when the transmitting device is unaware of the channel conditions (no CSI is available at the transmitter) and transmits the packet using fixed-rate and fixed transmission power. It leads to a high packet drop ratio and, consequently, a reduction in the QoS performance of the transmission link. Therefore, to enhance the QoS performance, different retransmission schemes such as automatic repeat request (ARQ) and hybrid-ARQ (HARQ) are used \cite{lin1984automatic}. These retransmission schemes provide better QoS performance at the expense of higher bandwidth utilization. Note that, when the channel conditions are bad, using retransmission schemes does not necessarily enhance QoS performance; however, it would surely lead to higher bandwidth utilization (due to many retransmissions). Moreover, when the transmitter knows the CSI before the transmission, the probability of packet drop is small. Therefore, using HARQ in this scenario is not as beneficial as when CSI at the transmitter is unknown. In short, we use the HARQ retransmission scheme (due to its benefits over ARQ in-terms of performance and bandwidth utilization) to enhance the QoS of the RIS-assisted D2D communication only when CSI is not available at the transmitting device. Like the previous sections, we use the EC to perform the throughput analysis of the HARQ-enabled RIS-assisted D2D communication under QoS constraints.

In HARQ-assisted D2D communication, the transmitting device transmits a packet on either the RIS-assisted D2D link or the RIS-assisted cellular link. If the receiver successfully decodes the received packet, it sends an acknowledgment, and the transmitting device transmits the next packet. On the other hand, if the receiver is unable to decode the received packet, it sends a negative acknowledgment and stores the erroneous packet in a buffer. The transmitting device then retransmits the same packet. The receiver uses the received packet from the current transmission attempt, as well as the packet stored in the buffer from the previous transmission attempt, to decode the packet. This process continues until either the receiver successfully decodes the received packet or an upper limit for retransmission attempts is reached. If the retransmission limit is reached, and the received packet is still not successfully decoded by the receiver, an outage occurs. Let $T$ be a transmission period containing $X$ copies of the data packet, and let $X$ be the upper limit of the retransmission attempts. Then, the decoding error probability of the $x^{\text{th}}$ attempt for RIS-assisted D2D and RIS-assisted cellular modes in underlay settings is \cite{li2017throughput},
\begin{equation}\label{decoding_error_prob}
\begin{split}
    P^{d}_{x}(h_1) &= Q \bigg( \frac{\sum_{n=1}^{x}\log_2(1+\Gamma_d^{^{\text{RIS}}}(n))+{\log(xl)/l} - r_t}{\log_2e\sqrt{\sum_{n=1}^{x}\frac{(2+\Gamma_d^{^{\text{RIS}}}(n))\Gamma_d^{^{\text{RIS}}}(n)}{l(\Gamma_d^{^{\text{RIS}}}(n) + 1)^2}}}\bigg)\\
    P^{c}_{x}(h_2) &= Q \bigg( \frac{\sum_{n=1}^{x}\log_2(1+\Gamma_c^{^{\text{RIS}}}(n))+{\log(xl)/l} - r_t}{\log_2e\sqrt{\sum_{n=1}^{x}\frac{(2+\Gamma_c^{^{\text{RIS}}}(n))\Gamma_c^{^{\text{RIS}}}(n)}{l(\Gamma_c^{^{\text{RIS}}}(n) + 1)^2}}}\bigg)
\end{split}
\end{equation}
where $x\in\{1,2,\dots,X\}$ and $l$ is the length of each fading block. We also assume that only one packet is being transmitted during transmission period $T$. Let $P_{t,m}$ be the probability of $m$ number of packets removed from the queue of $D_T$ in time $t$. Note that $m = 1$ due to the fact that a packet leaves the queue of $D_T$ at the end of $t$ either because of packet drop by the queue (in the case of decoding failure after $X$ unsuccessful attempts) or because of the successful decoding at $D_R$. Then, $P_{t,m}$ for RIS-assisted D2D and RIS-assisted cellular modes becomes,
\begin{equation}\label{eq:P_t_m}
  P_{t,m} =
  \begin{cases}
    \left.\begin{aligned}
        \E_{_{h_1}}[P^d_{t-1}]-\E_{_{h_1}}[P^d_{t}], \quad &t<X\\
        \E_{_{h_1}}[P^d_{_{X-1}}],\quad &t=X
       \end{aligned}
 \;\right\}
  \quad \text{\begin{tabular}[c]{@{}l@{}}RIS-assisted D2D mode\end{tabular}} \\
    \left.\begin{aligned}
        \E_{_{h_2}}[P^c_{t-1}]-\E_{_{h_2}}[P^c_{t}], \quad &t<X\\
        \E_{_{h_2}}[P^c_{_{X-1}}], \quad &t=X
       \end{aligned}
 \;\right\}
  \quad \text{\begin{tabular}[c]{@{}l@{}}RIS-assisted cellular mode\end{tabular}}\\
  \end{cases}
\end{equation}
Note that $t<X$ only corresponds to successful transmission. Note also that $t=X$ corresponds to either successful decoding at $X^{\text{th}}$ attempt or the packet getting dropped from the queue of $D_T$ as a result of an outage.

Now, to find the throughput under statistical QoS guarantees, we use the EC, which in this case becomes (Theorem 1 \cite{larsson2016effective}),
\begin{equation}\label{EC_HARQ}
    EC^{^{\text{RIS}}}_{H}(\varphi) = \frac{-1}{\varphi}\log_e \text{sp}(\mathbf{A})
\end{equation}
where $\mathbf{A}$ is the block companion matrix of size $X\times X$ and $\text{sp}(\mathbf{A})$ is the spectral radius (largest absolute eigenvalue) of matrix $\mathbf{A}$. Based on the number of retransmission attempts, matrix $\mathbf{A}$ can be defined as:
\begin{equation}\label{A}
\mathbf{A}=
\begin{bmatrix}
  a_1 & a_2 & \dots &a_{_{X-1}}& a_{_{X}} \\
 1 & 0 & \dots & 0 & 0 \\
  0 & 1 & \dots & 0 & 0 \\
  \vdots & \vdots & \ddots & \vdots & \vdots\\
  0 & 0 & \dots & 1 & 0
\end{bmatrix}.
\end{equation}
To find the entries of the matrix $\mathbf{A}$ ($a_1,a_2,\dots,a_{_{X-1}},a_{_{X}}$), we combine results of the Markov chain modelling of the D2D link and the probability of $m$ number of packets removed from the queue of $D_T$ in time $t$. Then, $a_{x}$ where $x\in\{1,2,\dots,X\}$ becomes,
\begin{equation}
	\label{eq:a_x}
a_x=
	 \begin{cases}   \mathbf{q}_{_{1}} \mathbf{\Phi}(-\varphi)\mathbf{p}_{u,i}^{\intercal}, &x=1\\
                     \mathbf{q}_{_{2}} \mathbf{\Phi}(-\varphi)\mathbf{p}_{u,i}^{\intercal}, &2\leq x \leq X-1 \\
                     \mathbf{q}_{_{3}} \mathbf{\Phi}(-\varphi)\mathbf{p}_{u,i}^{\intercal}, &x=X
                     \end{cases}
\end{equation}
where $\mathbf{p}_{u,i} = [p^u_1,p^u_2,p^u_3,p^u_4]$ is the vector containing state transition probabilities in underlay settings and where $\mathbf{\Phi}(-\varphi)$ is the diagonal matrix containing the LMGFs of four states given in Table \ref{states}. $\mathbf{q}_{_{1}}$, $\mathbf{q}_{_{2}}$, and $\mathbf{q}_{_{3}}$ are the vectors containing probabilities of $m$ number of packets removed from $D_T$'s queue in time $t$ for four Markov states. They can be written as the following: $\mathbf{q}_{_{1}} = [1-\E_{{h_1}}[P_1^d],1,1-\E_{{h_2}}[P_1^c],1]$, $\mathbf{q}_{_{2}} = [\E_{{h_1}}[P_{x-1}^d]-\E_{{h_1}}[P_{x}^d],1,\E_{{h_2}}[P_{x-1}^c]-\E_{{h_2}}[P_{x}^c],1]$, and $\mathbf{q}_{_{3}} = [\E_{{h_1}}[P_{_{X-1}}^d],1,\E_{{h_2}}[P_{_{X-1}}^c],1]$. By substituting these values in \eqref{eq:a_x}, we can find the entries of matrix $\mathbf{A}$. One can then find the spectral radius (largest absolute eigenvalue) of $\mathbf{A}$ by setting an upper limit of packet retransmissions ($X$). By substituting the largest absolute eigenvalue in \eqref{EC_HARQ}, one can also find the EC of HARQ-enabled RIS-assisted D2D communication in underlay settings.  

To find the EC of HARQ-enabled RIS-assisted D2D communication in overlay settings, one has to recompute the decoding error probabilities given in \eqref{decoding_error_prob} by replacing $\Gamma_d^{^{\text{RIS}}}(n)$ and $\Gamma_c^{^{\text{RIS}}}(n)$ with $\gamma_d^{^{\text{RIS}}}(n)$ and $\gamma_d^{^{\text{RIS}}}(n)$, respectively. These probability values, along with state-transition probabilities in overlay settings ($\mathbf{p}_{o,i} = [p^o_1,p^o_2,p^o_3,p^o_4]$), can then be used to compute the entries of the block-companion matrix $\mathbf{A}$. Similar to the computation of the EC in underlay settings, the EC in overlay settings can be calculated by substituting the largest absolute eigenvalue of $\mathbf{A}$ in \eqref{EC_HARQ}.

\section{Numerical Results}
The scope of this work is to propose an RIS-assisted D2D communication system for improving the D2D user's throughput while meeting the QoS requirements. In this section, we numerically illustrate the EC of an RIS-assisted D2D channel under statistical QoS constraints at the transmit device's queue.
\subsection{Simulation Setup}
We consider an RIS-assisted D2D network with two pairs of single-antenna user equipment that are positioned in a $300 m \times 300 m$ rectangular using a uniform distribution. One pair refers to the D2D pair ($D_T$ and $D_R$) and the other as the cellular user pair ($U_T$ and $U_R$). We assume that the LoS and the NLoS links of the channels $D_T \to$ RIS $\to D_R$, $D_R \to$ RIS $\to$ BS, and BS $\to$ RIS $\to D_R$ are indicated by the Rician fading channel. We set the Rician factor $\alpha = 4$ \cite{el2014modelling}. We set the number of RIS elements $N$ between $10$ to $100$.
\subsection{Simulation Results}
\begin{figure}[ht]
\begin{center}
	\includegraphics[width=4in]{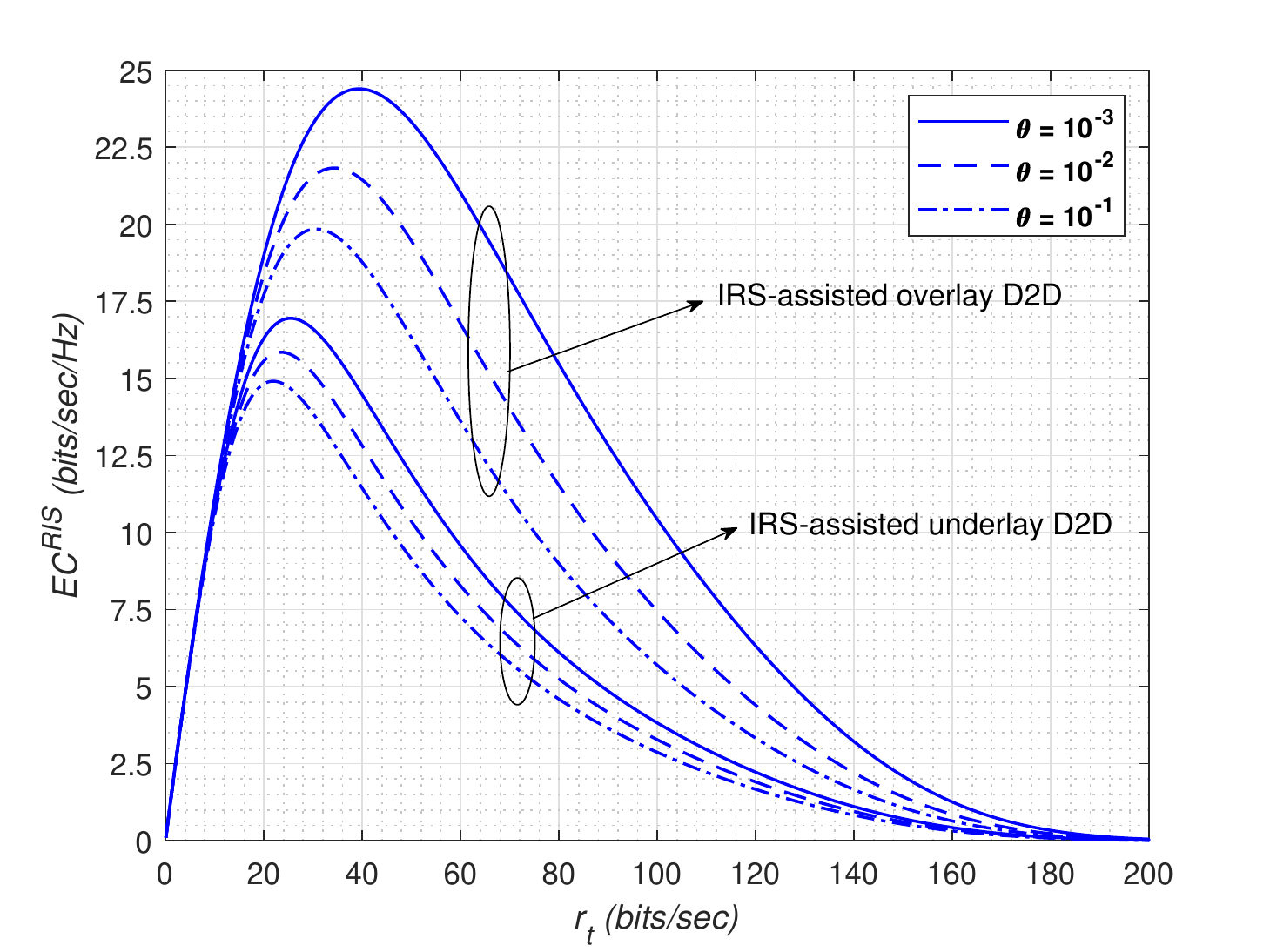}
\caption{EC vs transmission rate ($r_t$): A GD approach to find the optimal transmission rate to maximize the EC of RIS-assisted D2D communication without CSIT.}
\label{EC_vs_r}
\end{center}
\end{figure}
Fig. \ref{EC_vs_r} presents an exhaustive search through the GD algorithm to find the optimal transmission rates ($r_t$) in RIS-assisted overlay and underlay D2D communication when CSI is available at the receiver only. We observe that the EC of RIS-assisted D2D communication is a quasi-concave function of the transmission rate. For lower transmission rates, EC increases, and for higher values, it decreases. This numerical investigation reveals that when the transmission rate is below a certain threshold, EC increases with an increase in the transmission rate. It is because when the transmission rate is low, the low departure rate becomes the bottleneck for the EC. On the other hand, when the transmission rate is high, the outage probability (due to high packet drop ratio) becomes the bottleneck. It allows us to find the optimal transmission rate on which the D2D transmission link achieves maximum EC in the respective communication modes. We also observe that the optimal transmission rate changes with a change in the QoS constraints imposed at the transmit node's queue. A higher optimal transmission rate is achieved with a small QoS exponent and vice versa. It is because higher QoS constraints lead to higher outage probability and, consequently, low transmission rates. Additionally, one can also observe the influence of imposing stringent QoS constraints on the EC; for instance, a higher EC is obtained at the optimal transmission rate when loose QoS constraints are imposed at the transmit node's queue. Moreover, this simulation result studies the impact of D2D channel allocation on the EC. It can be seen that when the transmitter uses a dedicated D2D channel (overlay D2D mode), a better EC can be achieved when compared to the EC when the transmitter reuses the cellular user channel for transmission (underlay D2D mode). It is because the transmission in underlay mode experiences interference from the cellular user present in the network, and that causes the achievable EC to drop down. This can be verified from the analytical analysis presented in Section IV-A and IV-B.

\begin{figure}[ht]
\begin{center}
	\includegraphics[width=4in]{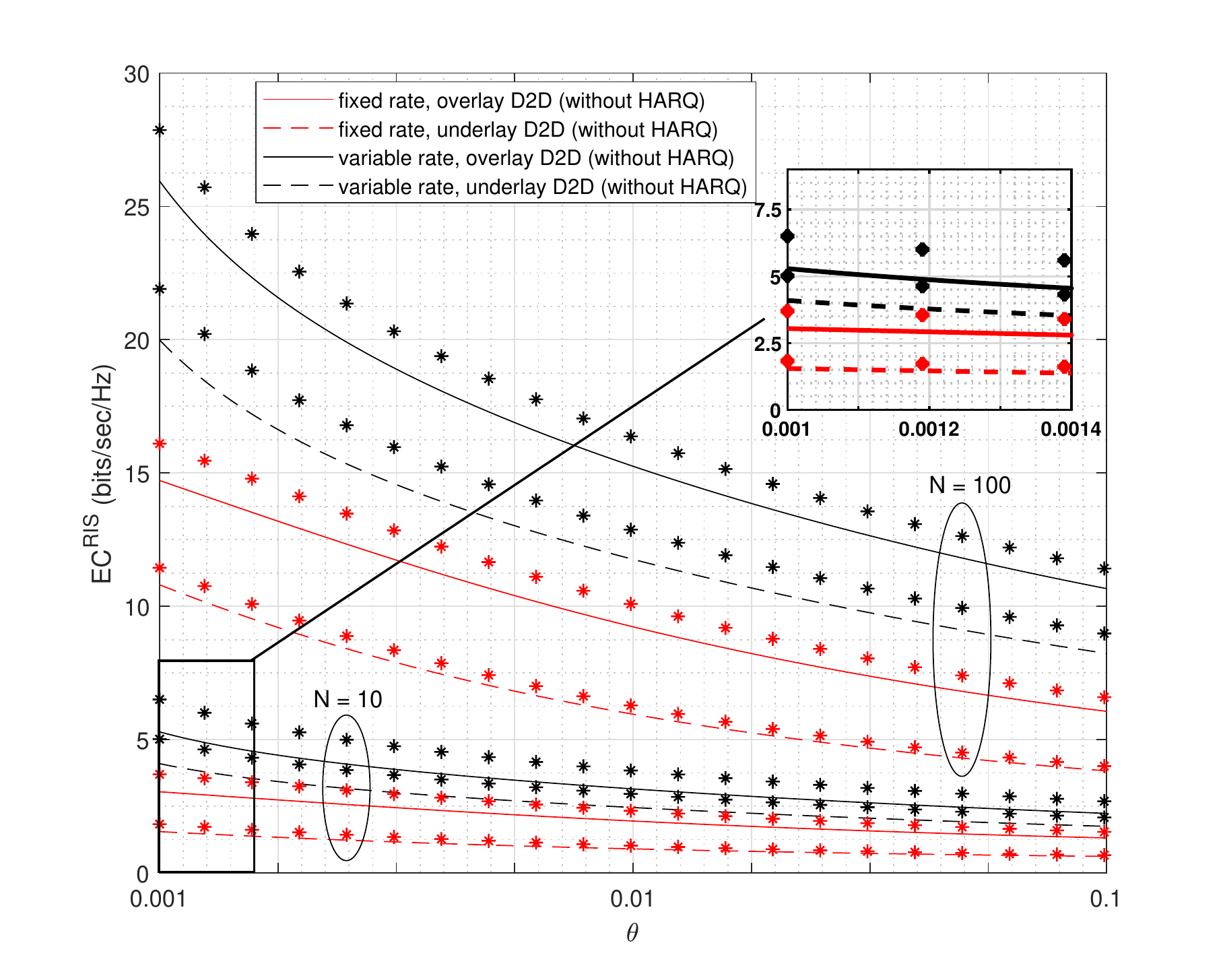}
\caption{EC vs QoS exponent ($\theta$): A comparison for EC of RIS-assisted D2D communication with and without CSIT and the impact of different number of RIS elements on the respective EC.}
\label{EC_vs_theta}
\end{center}
\end{figure}
In Fig. \ref{EC_vs_theta}, the EC of Rician block-fading RIS-assisted D2D channel for a fixed and variable transmission rate is plotted versus the QoS exponent ($\theta$) for different RIS elements. This simulation investigates the impact of having stringent QoS constraints imposed at the transmission queue on the EC of the RIS-assisted D2D communication link. The EC of RIS-assisted D2D communication decreases with an increase in $\theta$. We observe that the EC of RIS-assisted D2D link when CSIT is available is higher than the EC with no CSIT. It is because when the transmitter is aware of CSI prior to the transmission, it adjusts its transmission rate according to the channel conditions; hence, a better EC can be achieved. On the other hand, the transmitter sends data using a fixed transmission rate when CSIT is unknown. However, this gain of transmission with CSIT decreases as stricter QoS constraints are imposed at the transmission queue. We also observe that higher EC can be achieved when the RIS-assisted D2D link operates in overlay mode when compared to underlay mode. Section IV and V provides the complete analytical analysis for calculating the EC of the RIS-assisted D2D communication in underlay and overlay communication modes with and without instantaneous CSIT, respectively. This simulation result also studies the impact of increasing the number of RIS elements ($N$) on the achievable EC in the respective communication modes. We note that as $N$ increases, the EC of RIS-assisted D2D link also increases. We observe at least five times better EC is achieved with $N = 100$ as compared to EC with $N = 10$. It shows the efficacy of using a large RIS. However, increasing the number of RIS elements also increases the hardware cost. 

\begin{figure}[ht]
\begin{center}
	\includegraphics[width=4in]{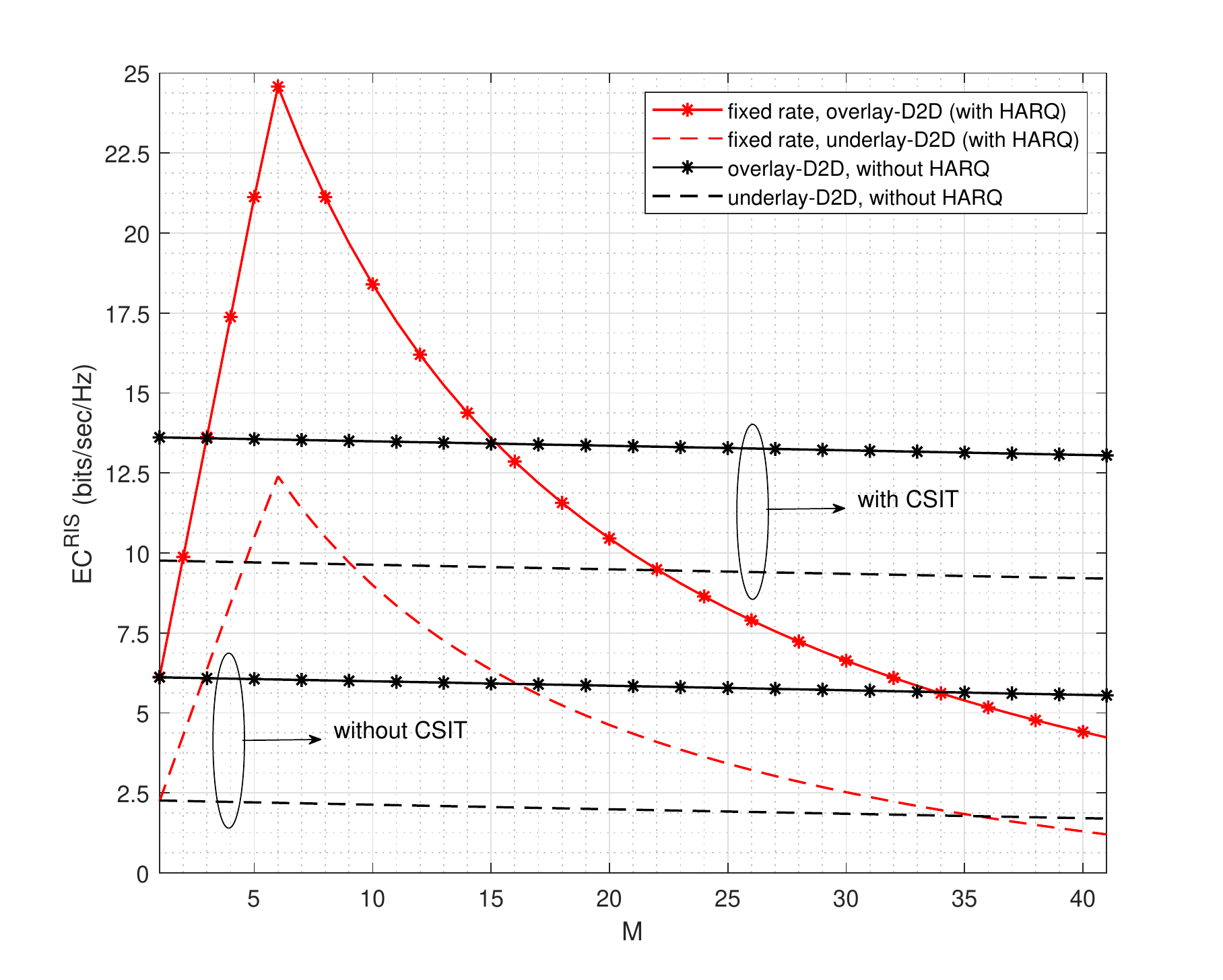}
\caption{Impact of HARQ on the EC of RIS-assisted D2D communication: EC vs retransmission deadline constraint ($M$).}
\label{EC_vs_M}
\end{center}
\end{figure}
Fig. \ref{EC_vs_M} presents the numerical investigation for the analytical findings of Section VI, which provides a comprehensive framework for enhancing the QoS performance of the RIS-assisted D2D communication link when instantaneous CSIT is unknown. In this case, we leverage the HARQ retransmission scheme to achieve better EC of the candidate D2D link in both underlay and overlay communication modes; hence, an improved QoS performance. Fig. \ref{EC_vs_M} provides an analysis to understand the impact of having multiple retransmissions of a packet on the achievable EC of the RIS-assisted D2D communication link. We plot the EC under fixed and variable transmission rate as a function of the retransmission limit $M$. We observe that an optimal value of $M$ exists that maximizes the EC of HARQ-enabled RIS-assisted D2D link. The EC for this case increases with an early increase in $M$, and after a unique value of $M$ is reached (i-e. $M=6$), it starts decreasing. It is because, for smaller $M$, the transmit node has to reduce the transmission rate to meet the target outage probability, which results in reduced EC. Similarly, the transmission rate increases with an increase in $M$, which consequently improves the EC of the RIS-assisted D2D link. After the EC reaches its maximum value under a specific QoS constraint, a further increase in $M$ reduces the EC. Note that the HARQ retransmission scheme is used when no CSIT is available and the transmit node uses a fixed transmission rate. Therefore, for $M=1$, the EC of HARQ-enabled D2D is equal to the EC when no retransmission scheme is used, and the transmit node uses a fixed transmission rate. This numerical investigation reveals that four times better EC can be achieved using HARQ with optimal $M$ as compared to EC of the same link without HARQ.

\begin{figure}[ht]
\begin{center}
	\includegraphics[width=4in]{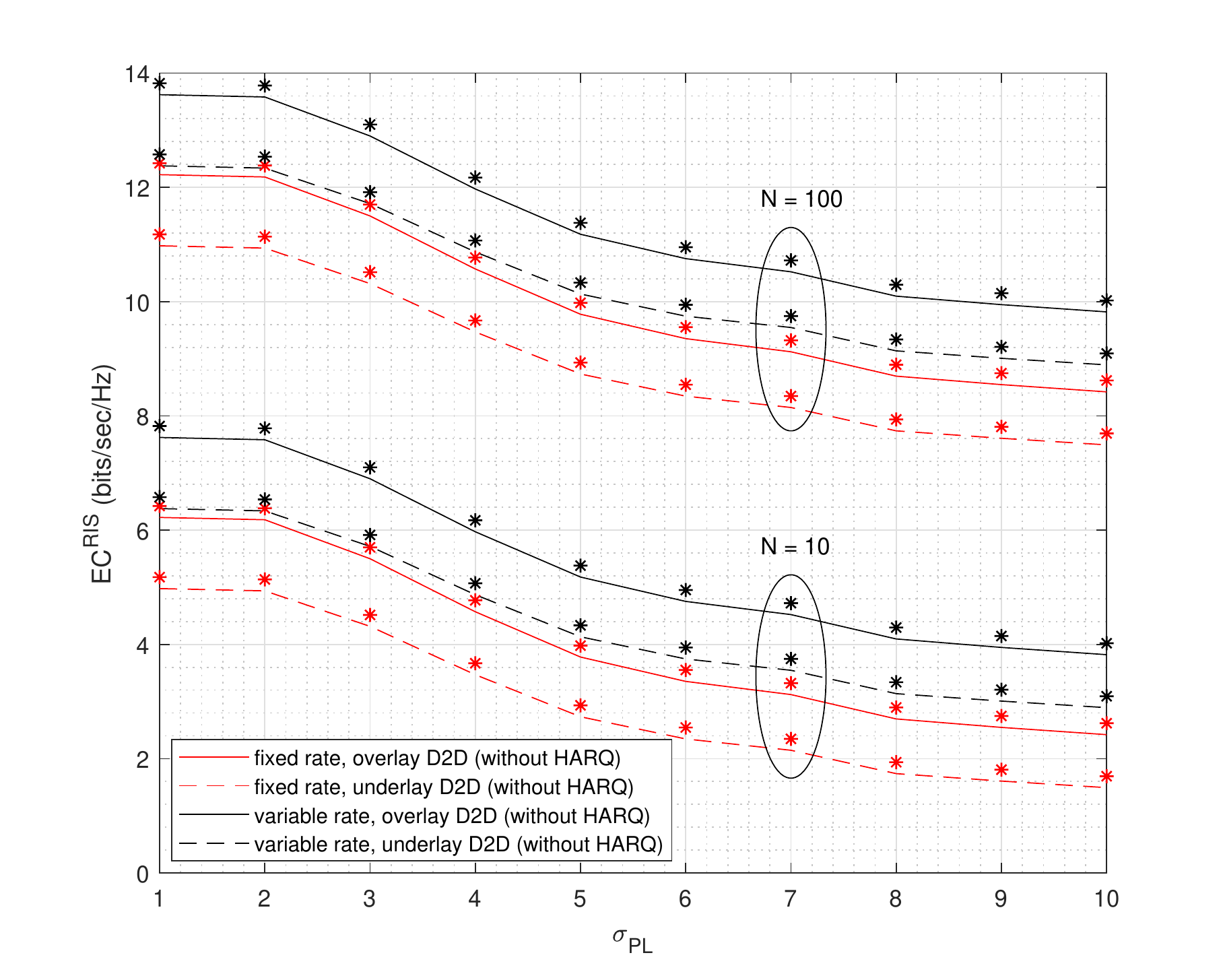}
\caption{Impact of mode selection on the EC of RIS-assisted D2D communication: EC vs variance of the estimation error of the PL measurements ($\sigma_{_{PL}}$) for fixed and variable transmission rates.}
\label{EC_vs_sigma}
\end{center}
\end{figure}
Next, we investigate the impact of the mode selection mechanism (presented in Section III) on the EC of RIS-assisted D2D link in Fig. \ref{EC_vs_sigma}. The EC decreases with an increase in the variance of the estimation error of path loss measurements ($\sigma_{_{PL}}^2$). It is because the quality of the estimation techniques used for estimating the path loss measurements for the mode selection has a direct impact on the system's performance. As the quality decreases, the variance of the estimation error increases, which consequently reduces the EC. It shows the paramount importance of designing efficient estimators for the path loss measurements for mode selection. Moreover, one can also see the impact of using different sets of RIS elements, overlay and underlay D2D, and fixed and variable transmission rates on the EC versus $\sigma_{_{PL}}^2$. 

\begin{figure}[ht]
\begin{center}
	\includegraphics[width=4in]{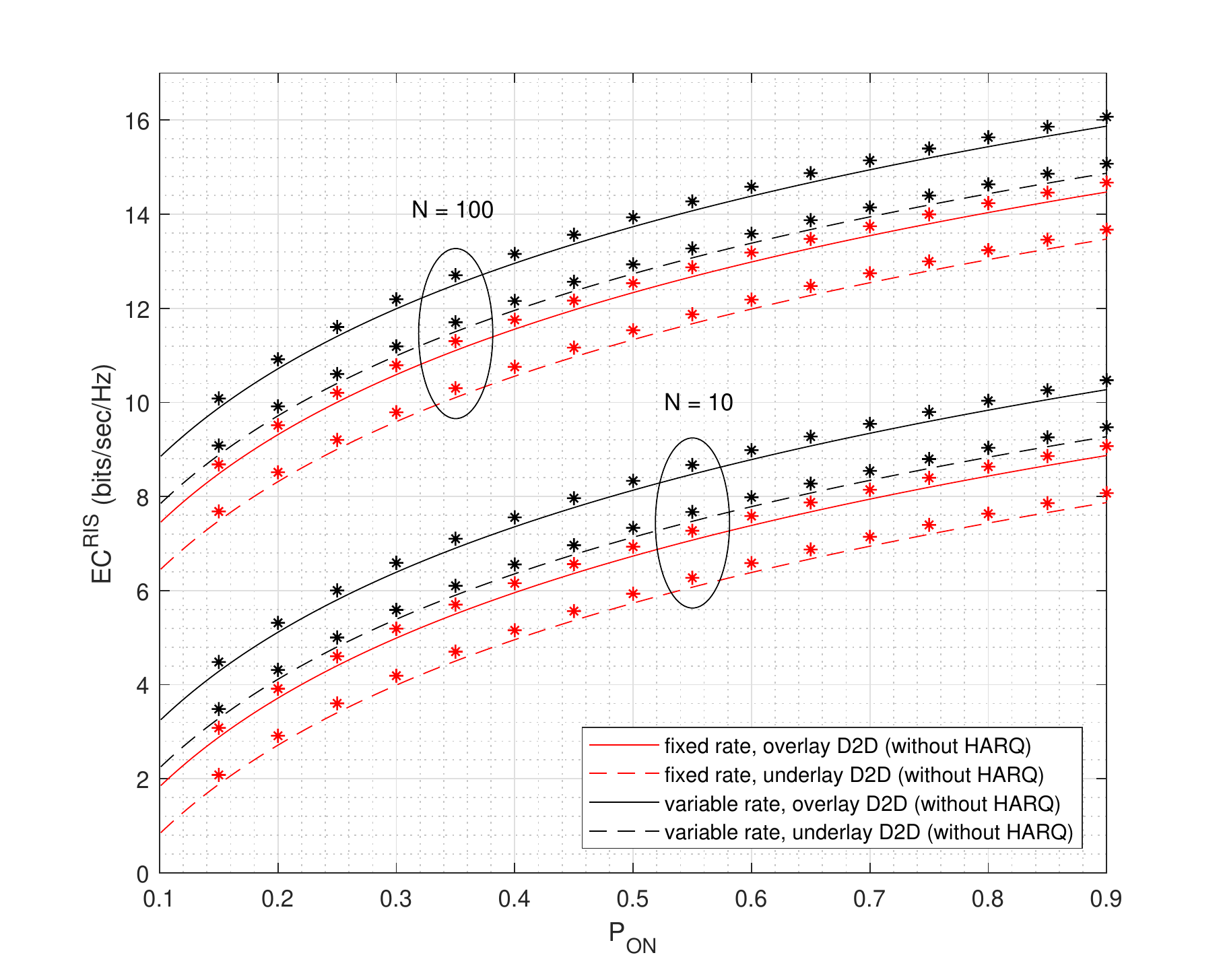}
\caption{EC of RIS-assisted D2D communication vs the probability of ON states ($P_{ON}$) for fixed and variable transmission rates.}
\label{EC_vs_PON}
\end{center}
\end{figure}
Fig. \ref{EC_vs_PON} presents the impact of Markov chain modeling of the D2D link on the EC of RIS-assisted D2D communication. The EC, with and without the knowledge of CSIT, is plotted versus the ON states probability. $P_{_{ON}}$ is the accumulative probability of ON states. For fixed-rate transmission, $P_{_{ON}} = p_1^u+p_3^u$ and $P_{_{ON}} = p_1^o+p_3^o$ for underlay and overlay modes,  respectively. For variable rate transmission, $P_{_{ON}} = p_1+p_3+p_4$ for both overlay and underlay modes. We observe that the EC increases with an increase in $P_{_{ON}}$. It is because, when the Markov D2D channel is in ON state, it transmits with $r_t$ and $r_t(n)$ in fixed and variable rate transmission, respectively, which leads to higher EC. On the other hand, when the Markov D2D channel is in OFF state, it does not transmit at all; thus, lowering the EC. In short, if the probability of the D2D channel is high for being in ON state, it increases the EC due to continuous transmission. 

\section{Conclusion}
In this paper, we have studied the throughput of the RIS-assisted D2D communication subject to delay QoS constraints. We modeled the RIS-assisted D2D link as a Markov service process with transition probabilities and formulated the closed-form expression for the effective capacity (EC) using these probabilities. We have extended our analysis to both cases when the channel state information at the transmitter (CSIT) is available and not. We observed that the packet drop ratio increases when the transmit D2D node transmits without prior knowledge of CSIT. Therefore, we used the hybrid automatic repeat request (HARQ) retransmission scheme to enhance the throughput of the RIS-assisted D2D link. Through simulation results, we have observed five times better EC can be achieved using a higher number of RIS elements ($N$). However, increasing the number of RIS elements also increases the hardware cost. Thus, the size of the RIS should be subjected to the link's throughput requirements. We also observed that there exists an optimal transmission rate ($r_t$) and an optimal retransmission deadline constraint ($M$) that maximizes the EC of the RIS-assisted D2D link. Last but not least, we also observed that the EC decreases with an increase in the variance of the path loss measurement error. 
\appendices
\section{Phase Shift Designs for RIS Elements}\label{prop_phase_shift}
Channel of the RIS-assisted D2D link ($D_T \to$ RIS $\to D_R$) can be written as $h^{^{D_T}}_{_{\textnormal{RIS}}}\mathbf{\Theta}h_{_{D_R}}^{^{\textnormal{RIS}}}$, where $h^{^{D_T}}_{_{\textnormal{RIS}}}$ and $h_{_{D_R}}^{^{\textnormal{RIS}}}$ are the channel coefficients of $D_T \to \textnormal{RIS}$ and $\textnormal{RIS} \to D_R$ link, respectively, and $\mathbf{\Theta}$ is the phase shift matrix. Calculation of the phase shift matrix critically depends on the availability of the CSI at the transmitter before the transmission happens. As mentioned above, we consider two cases in this work, first when $D_T$ sends data without knowing the instantaneous CSI (average CSI instead), and second when the transmitter has the perfect instantaneous CSI prior to the transmission. Therefore, we present the framework for calculating the optimal phase shift matrix for both the cases. 
\begin{itemize}
    \item \textbf{Phase Shift Design based on Statistical CSI (Path Loss):}
    The optimal phase shift matrix when the transmitter sends data knowing only statistical CSI can be written as,
    \begin{equation}
        \mathbf{\Theta}^{^{\text{stat}}}_{_{\text{opt}}} = \text{diag} \bigg( \big[ e^{j\phi_{1,\text{opt}}^{\text{stat}}}, e^{j\phi_{2,\text{opt}}^{\text{stat}}}, \dots , e^{j\phi_{N,\text{opt}}^{\text{stat}}}  \big]^T \bigg).
    \end{equation}
    where $e^{j\phi_{1,\text{opt}}^{\text{stat}}}, e^{j\phi_{2,\text{opt}}^{\text{stat}}}, \dots , e^{j\phi_{N,\text{opt}}^{\text{stat}}}$ for $n\in (1,2,\dots,N)$ are the optimal phase shift values for each RIS element when only the statistical CSI is known at the transmitter. Phase shift value of each RIS element can be found using the following expression. 
    \begin{equation}
        \phi_{n,\text{opt}}^{\text{stat}} = -\text{arg} \big(\bar{h}^{^{D_T}}_{_{\text{RIS}}}(n) \big)-\text{arg} \big(\bar{h}_{_{D_R}}^{^{\text{RIS}}}(n)  \big).
    \end{equation}
    where $\bar{h}^{^{D_T}}_{_{\text{RIS}}}(n)$ and $\bar{h}_{_{D_R}}^{^{\text{RIS}}}(n)$ are the average values of the channel coefficients for $n^{th}$ RIS element. It is important to note here that the statistical CSI (path loss) varies slowly in the wireless channel, and once estimated, can last for multiple time intervals. Therefore, $ \mathbf{\Theta}^{^{\text{stat}}}_{_{\text{opt}}}$ is useful for a longer period of time. 
    \item \textbf{Phase Shift Design based on Instantaneous Perfect CSI:}
    The optimal phase shift matrix when the transmitter knows the instantaneous perfect CSI prior to sending data can calculated using the following expression. 
    \begin{equation}
        \mathbf{\Theta}^{^{\text{inst}}}_{_{\text{opt}}} = \text{diag} \bigg( \big[ e^{j\phi_{1,\text{opt}}^{\text{inst}}}, e^{j\phi_{2,\text{opt}}^{\text{inst}}}, \dots , e^{j\phi_{N,\text{opt}}^{\text{inst}}}  \big]^T \bigg),
    \end{equation}
    with each instance becomes,
      \begin{equation}
        \phi_{n,\text{opt}}^{\text{inst}} = -\text{arg} \big(h^{^{D_T}}_{_{\text{RIS}}}(n) \big)-\text{arg} \big(h_{_{D_R}}^{^{\text{RIS}}}(n)  \big).
    \end{equation}
In contrast to statistical CSI, instantaneous CSI changes rapidly for a wireless channel due to the small-scale fading. Even if the wireless channel is considered stationary, the small-scale fading needs to be estimated multiple times in a time interval. Therefore, the optimal phase shifts of the RIS elements should be updated frequently to adapt to the channel conditions that characterizes the small-scale fading \cite{jung2020performance}. 
\end{itemize}

\section{Proof of Proposition 4.1}\label{prop1}
The outage probability of RIS-assisted D2D link ($D_T \to$ RIS $\to D_R$) is $P[\Psi_d^{^{\text{RIS}}}(n)<\gamma_{_{T}}]$, where $\Psi_d^{^{\text{RIS}}}(n)$ is the SIR of the link. Let
\begin{equation*}
\psi_d^{^{\text{RIS}}}=\frac{\bar{P}_{_{D_T}}\big(\sum_{m_z,m_y}h^{^{D_T,D_R}}_{m_z,m_y}(n)e^{j\phi_{m_z,m_y}}\big)^2}{PL_{d}}
\end{equation*}
and
\begin{equation*}
I_d=\frac{\bar{P}_{_{U_T}}\big(h_{_{U_T,D_R}})^2}{PL_{_{U_T,D_R}}}.
\end{equation*}
Then, $\Psi_d^{^{\text{RIS}}} = \psi_d^{^{\text{RIS}}}\big/I_d$. Note that $\psi_d^{^{\text{RIS}}}$ and $I_d$ are exponentially distributed random variables (RVs); therefore, $\psi_d^{^{\text{RIS}}} \sim \exp(\alpha_1)$ and $I_d \sim \exp(\alpha_2)$, where $\alpha_1 = \E[\psi_d^{^{\text{RIS}}}]$ and $\alpha_2 = \E[I_d]$ are the mean of the respective distributions. In finding the mean of the distributions, we assume that there is no phase error experienced by the reflected signals (i-e. $\phi_{m_z,m_y} = 0$). Thus, $\alpha_1 = \frac{PL_d}{N\pi \bar{P}_{_{D_T}}}$ (Appendix A \cite{aman2020effective}) and $\alpha_2 = \frac{PL_{_{U_T,D_R}}}{ \bar{P}_{_{U_T}}}$. The cumulative distribution function (CDF) of $\Psi_d^{^{\text{RIS}}}$ then becomes: $P[\Psi_d^{^{\text{RIS}}}(n)<\gamma_{_{T}}] = \frac{\alpha_1}{(\alpha_1 + \alpha_2)/\gamma_{_{T}}}$. Now, by substituting in the values of $\alpha_1$ and $\alpha_2$, and after some steps of simplifications, the final result becomes:
\begin{equation}
  P[\Psi_d^{^{\text{RIS}}}(n)<\gamma_{_{T}}] = \frac{\bar{P}_{_{U_T}}PL_{d}\gamma_{_{T}}}{PL_{d}\bar{P}_{_{U_T}} + PL_{_{U_T,D_R}}\bar{P}_{_{D_T}}N\pi}.
\end{equation}
A similar outage probability analysis is done in \cite{lee2010outage}, in which the authors first propose a relay selection mechanism and then use it to derive the outage probability of a cognitive relay network. Their results investigate the impact of the distance ratio of the interference link to the relaying link on the outage performance of cognitive relay networks.
\section{Proof of Proposition 4.2}\label{prop2}
The outage probability of the RIS-assisted cellular link is $P[\Psi_c^{^{\text{RIS}}}(n)<\gamma_{_{T}}]$, where $\Psi_c^{^{\text{RIS}}}(n)$ is the net SIR of the uplink and the downlink channels. This net SIR can be calculated as: $\Psi_c^{^{\text{RIS}}}(n) = \min\{\Psi_{ul}^{^{\text{RIS}}}(n), \Psi_{dl}^{^{\text{RIS}}}(n) \}$, where $\Psi_{ul}^{^{\text{RIS}}}(n)$ and $\Psi_{dl}^{^{\text{RIS}}}(n)$ are the SIR of the uplink ($D_T \to$ RIS $\to$ BS) and the downlink (BS $\to$ RIS $\to D_R$) channels, respectively. Let,
\begin{equation*}
\begin{split}
\psi_{ul}^{^{\text{RIS}}} &=\frac{\bar{P}_{_{D_T}}\big(\sum_{m_z,m_y}h^{^{D_T,\text{BS}}}_{m_z,m_y}(n)e^{j\phi_{m_z,m_y}}\big)^2}{PL_{_{D_T,\text{BS}}}}\\
\psi_{dl}^{^{\text{RIS}}} &=\frac{\bar{P}_{_{\text{BS}}}\big(\sum_{m_z,m_y}h^{^{\text{BS},D_R}}_{m_z,m_y}(n)e^{j\phi_{m_z,m_y}}\big)^2}{PL_{_{\text{BS},D_R}}}\\
I_{ul} &=\frac{\bar{P}_{_{U_T}}\big(h_{_{U_T,\text{BS}}}\big)^2}{PL_{_{U_T,\text{BS}}}}.
\end{split}
\end{equation*}
Then $\Psi_{ul}^{^{\text{RIS}}} = \psi_{ul}^{^{\text{RIS}}}\big/ I_{ul}$ and $\Psi_{dl}^{^{\text{RIS}}}(n) = \psi_{dl}^{^{\text{RIS}}}\big/ I_{d}$, where $I_{d}$ is similar to the one in Appendix \ref{prop1}. Similar to $I_{d}$, we observe that $\psi_{ul}^{^{\text{RIS}}}$, $\psi_{dl}^{^{\text{RIS}}}$, and $I_{ul}$ are also exponentially distributed RVs. Therefore, $\psi_{ul}^{^{\text{RIS}}} \sim \exp(\beta_1)$, $\psi_{dl}^{^{\text{RIS}}}\sim \exp(\beta_2)$, and $I_{ul}\sim \exp(\beta_3)$, where $\beta_1 = \E[\psi_{ul}^{^{\text{RIS}}}]$, $\beta_2 = \E[\psi_{dl}^{^{\text{RIS}}}]$, and $\beta_3 = \E[I_{ul}]$ are the means of the respective distributions. Similar to Appendix \ref{prop1}, we assume that there is no phase error experienced by the reflected signals; then, $\beta_1 = \frac{PL_{_{D_T,\text{BS}}}}{N\pi \bar{P}_{_{D_T}}}$, $\beta_2 = \frac{PL_{_{\text{BS},D_R}}}{N\pi \bar{P}_{_{\text{BS}}}}$, and $\beta_3 = \frac{PL_{_{U_T,\text{BS}}}}{ \bar{P}_{_{U_T}}}$. Because $\Psi_{ul}^{^{\text{RIS}}}$ and $\Psi_{dl}^{^{\text{RIS}}}$ are independent RVs, the CDF of $\Psi_c^{^{\text{RIS}}}$ becomes: 
\begin{equation}\label{outage_prop2}
\begin{split}
    P[\Psi_c^{^{\text{RIS}}}(n)<\gamma_{_{T}}] =& \frac{\beta_1}{(\beta_1+\beta_3)/\gamma_{_{T}}} + \frac{\beta_2}{(\beta_2+\alpha_2)/\gamma_{_{T}}}- \frac{\beta_1}{(\beta_1+\beta_3)/\gamma_{_{T}}} \text{x} \frac{\beta_2}{(\beta_2+\alpha_2)/\gamma_{_{T}}}.
\end{split}
\end{equation} 
Now, by utilizing the values of $\beta_1$, $\beta_2$, $\beta_3$, and $\alpha_2$ (from Appendix \ref{prop1}), and after some simplification steps, each term becomes:
\begin{equation*}
\begin{split}
\frac{\beta_1}{(\beta_1+\beta_3)/\gamma_{_{T}}} &= \frac{\bar{P}_{_{U_T}}PL_{_{D_T,\text{BS}}}\gamma_{_{T}}}{\bar{P}_{_{U_T}}PL_{_{D_T,\text{BS}}} + \bar{P}_{_{D_T}}PL_{_{U_T,\text{BS}}}N\pi}\\
\frac{\beta_2}{(\beta_2+\alpha_2)/\gamma_{_{T}}} &= \frac{\bar{P}_{_{U_T}}PL_{_{\text{BS},D_R}}\gamma_{_{T}}}{\bar{P}_{_{U_T}}PL_{_{\text{BS},D_R}} + \bar{P}_{_{\text{BS}}}PL_{_{U_T,D_R}}N\pi}.
\end{split}
\end{equation*}
By substituting values of each term in \eqref{outage_prop2}, and after some simplification steps, the outage probability of the RIS-assisted cellular link is as follows:
\begin{equation}
\begin{split}
&P[\Psi_c^{^{\text{RIS}}}(n)<\gamma_{_{T}}] =\frac{\bar{P}_{_{U_T}}\gamma_{_{T}}\big[PL_{_{D_T,\text{BS}}}\big\{\bar{P}_{_{U_T}}PL_{_{\text{BS},D_R}}(2-\gamma_{_{T}})+\Omega_1\big\}-PL_{_{\text{BS},D_R}}\Omega_2\big]}{(\bar{P}_{_{U_T}}PL_{_{D_T,\text{BS}}} + \Omega_1)(\bar{P}_{_{U_T}}PL_{_{\text{BS},D_R}} +\Omega_2)}.    
\end{split}
\end{equation}
Where $\Omega_1 =  \bar{P}_{_{D_T}}PL_{_{U_T,\text{BS}}}N\pi$ and $\Omega_2 =  \bar{P}_{_{\text{BS}}}PL_{_{U_T,D_R}}N\pi$.
\section{Proof of Proposition 4.3}\label{prop3}
The RIS-assisted cellular link is a two-hop wireless link consisting of an uplink ($D_T \to$ RIS $\to$ BS) and a downlink (BS $\to$ RIS $\to D_R$) channels. Therefore, to find the mean of the net SNR of the RIS-assisted cellular link, we have to find the means of the SNR of uplink and downlink channels. Let $\kappa_{_{ul}} = \E[\gamma_{ul}^{^{\text{RIS}}}(n)]$ and $\kappa_{_{dl}} = \E[\gamma_{dl}^{^{\text{RIS}}}(n)]$ represent the means of the SNR of the uplink and the downlink channels, respectively. Then, by assuming zero phase error for the reflected signals ($\phi_{m_z,m_y} = 0$), these mean values can be written as (Appendix A \cite{aman2020effective}): $\kappa_{ul} = \E[\gamma_{ul}^{^{\text{RIS}}}(n)] = N\pi\bar{P}_{_{D_T}}\big/PL_{_{D_T,\text{BS}}}\omega_{_{0}}$
and $\kappa_{dl} = \E[\gamma_{dl}^{^{\text{RIS}}}(n)] = N\pi\bar{P}_{_{\text{BS}}}\big/PL_{_{\text{BS},D_R}}\omega_{_{0}}$.

Note that $\gamma_c^{^{\text{RIS}}}(n) = \min \{\gamma_{ul}^{^{\text{RIS}}}(n),\gamma_{dl}^{^{\text{RIS}}}(n) \}$. Because $\gamma_{ul}^{^{\text{RIS}}}(n)$ and $\gamma_{dl}^{^{\text{RIS}}}(n)$ are exponentially distributed RVs and because the minimum of two RVs is an exponential RV, $\gamma_{c}^{^{\text{RIS}}}(n)$ is also an exponential RV. Let $\kappa_c = \E[\gamma_c^{^{\text{RIS}}}(n)]$ be the mean of the net SNR of the RIS-assisted cellular link. It can then be written as:
\begin{equation}\label{kappa_c}
\kappa_c = \E[\gamma_c^{^{\text{RIS}}}(n)] = \frac{\kappa_{_{ul}} \kappa_{_{dl}}}{\kappa_{_{ul}}+ \kappa_{_{dl}}}.
\end{equation}
By substituting $\kappa_{_{ul}}$ and $\kappa_{_{dl}}$ in \eqref{kappa_c}, the numerator of \eqref{kappa_c} becomes,
\begin{equation*}
num = \kappa_{_{ul}} \kappa_{_{dl}} = \frac{N^2\pi^2\bar{P}_{_{D_T}}\bar{P}_{_{\text{BS}}}}{PL_{_{D_T,\text{BS}}}PL_{_{\text{BS},D_R}}\omega_{_{0}}^2}.
\end{equation*}
Similarly, the denumerator of \eqref{kappa_c} becomes,
\begin{equation*}
denum = \kappa_{_{ul}}+\kappa_{_{dl}}=\frac{N\pi(\bar{P}_{_{D_T}}PL_{_{\text{BS},D_R}} + \bar{P}_{_{\text{BS}}}PL_{_{D_T,\text{BS}}})}{PL_{_{D_T,\text{BS}}}\omega_{_{0}}PL_{_{\text{BS},D_R}}}.
\end{equation*}
By substituting $num$ and $denum$ in \eqref{kappa_c}, and after some simplification steps, the mean of the net SNR of of the RIS-assisted cellular link ($\E[\gamma_c^{^{\text{RIS}}}(n)]$) becomes:
\begin{equation}
\kappa_c = \E[\gamma_c^{^{\text{RIS}}}(n)] = \frac{N\pi\bar{P}_{_{D_T}}\bar{P}_{_{\text{BS}}}}{\bar{P}_{_{D_T}}PL_{_{\text{BS},D_R}}+\bar{P}_{_{\text{BS}}}PL_{_{D_T,\text{BS}}}}.
\end{equation}
\footnotesize{
\bibliographystyle{IEEEtran}
\bibliography{references}
}

\vfill\break

\end{document}